\documentclass{article}   
\usepackage{latexsym}     
\usepackage{ams}  
\usepackage{amsthm}     
\usepackage{amssymb}     
\usepackage{amsmath}     
\usepackage{proof}
\usepackage{graphics}     
\usepackage{color}     
\usepackage[all]{xy}    
\newtheorem{theorem}{Theorem}[section]
\newtheorem{lemma}[theorem]{Lemma}
\newtheorem{claim}[theorem]{Claim}
\newtheorem{corollary}[theorem]{Corollary} 
\newtheorem{proposition}[theorem]{Proposition}

\theoremstyle{definition}

\newcommand{\FO}{\mbox{FO}}
\newcommand{\bbP}{\ensuremath{\mathbb{P}}}  
\newcommand{\bP}{\ensuremath{\mathbf{P}}}  
\newcommand{\bR}{\ensuremath{\mathbf{R}}}  
\newcommand{\bC}{\ensuremath{\mathbf{C}}}  
\newcommand{\cF}{\ensuremath{\mathcal{F}}}  
\newcommand{\cS}{\ensuremath{\mathcal{S}}}  
\newcommand{\cR}{\ensuremath{\mathcal{R}}}  
\newcommand{\fA}{\ensuremath{\mathfrak{A}}}  
\newcommand{\fB}{\ensuremath{\mathfrak{B}}}  
\newcommand{\fC}{\ensuremath{\mathfrak{C}}}  
\newcommand{\fM}{\ensuremath{\mathfrak{M}}}  
\newcommand{\semantics}[1]{#1^{\fA}}
\renewcommand{\phi}{\varphi}
\newcommand{\X}{{\sf X}}
\renewcommand{\R}{{\sf R}} 
\renewcommand{\S}{{\sf S}} 

\newcommand{\dbar}{\bar{d}}

\newcommand{\mbar}{\bar{m}}
\newcommand{\nbar}{\bar{n}}
\newcommand{\proves}{\vdash}
\newcommand{\genProves}{\mid \hspace{-0.06cm}\sim}

\newcommand{\rbar}{\bar{r}}
\newcommand{\obar}{\bar{o}}
\newcommand{\pbar}{\bar{p}}

\newcommand{\rem}[1]{\relax}
\newcommand{\set}[1]{\{ #1 \}}
\newcommand{\pair}[1]{\langle #1 \rangle}

\newcommand{\NLOGSPACE}{\textsc{NLogSpace}}
\newcommand{\NEXPTIME}{\textsc{NExpTime}}
\newcommand{\EXPTIME}{\textsc{ExpTime}}
\newcommand{\PTIME}{\textsc{PTime}}
\newcommand{\NPTIME}{\textsc{NPTime}}

\title{Logics for  the Relational Syllogistic} \author{
Ian Pratt-Hartmann \and Lawrence S.~Moss} \date{}
\begin{document}
\maketitle
\begin{abstract}
\noindent
The Aristotelian syllogistic cannot account for the validity of many
inferences involving relational facts.  In this paper, we investigate
the prospects for providing a relational syllogistic.  We identify
several fragments based on (a) whether negation is permitted on all
nouns, including those in the subject of a sentence; and (b) whether
the subject noun phrase may contain a relative clause.  The logics we
present are extensions of the classical syllogistic, and
we pay special attention to the question of whether \emph{reductio ad
  absurdum} is needed.  Thus our main goal is to derive results on the
existence (or non-existence) of syllogistic proof systems for
relational fragments.  We also determine the computational complexity
of all our fragments.
\end{abstract}

\tableofcontents

\section{Introduction}
\label{sec:introduction}
Augustus de Morgan famously observed that the Aristotelian syllogistic
cannot account for the validity of even the most elementary inferences
involving relational facts, for example (de
Morgan~\cite{logic:deMorgan47}, p.~114):
\begin{equation}
\mbox{
\begin{minipage}{10cm}
\begin{tabbing}
\underline{\sf Every man is an animal \hspace{1.5cm}}\\
{\sf He who kills a man kills an animal}.
\end{tabbing}
\end{minipage}
}
\label{eq:arg1}
\end{equation}
De Morgan was certainly not the first to notice the problem of
relational inference: for example, it is given prominent treatment in
the Port Royal Logic (Arnauld~\cite{logic:arnauld62}, Ch.~III). But
whereas the Port-Royalists took such inferences to demonstrate at most
the occasional need for ingenious reformulation, de Morgan saw in them
clear evidence that binary relations must be governed by logical
principles lying outside the scope of the traditional syllogistic, and
made one of the first serious attempts to extend syllogism-like
principles to relational judgments (see de
Morgan~\cite{deMorgan60}). Ultimately, of course, de Morgan's attempt
was to be without issue: the methods he employed were rendered
obsolete at the end of the nineteenth century by the rise of predicate
logic, which provided an expressive yet easily-understood apparatus
for formalizing relational facts.  (For a modern reconstruction of de
Morgan's work, see Merrill~\cite{logic:merrill90}.)

From a computational point of view, however, expressive power is a
double-edged sword: roughly speaking, the more expressive a language
is, the harder it is to compute with. In the last decade, this
trade-off has led to renewed interest in {\em inexpressive} logics, in
which the problem of determining entailments is algorithmically
decidable with (in ideal cases) low complexity. The logical fragments
subjected to this sort of complexity-theoretic analysis have naturally
enough tended to be those which owe their salience to the syntax of
first-order logic, for example: the two-variable fragment, the guarded
fragment, and various quantifier-prefix fragments. But of course it is
equally reasonable to consider instead logics defined in terms of the
syntax of {\em natural languages}. Perhaps, therefore, it is time to
return to where de Morgan and his contemporaries left off, but with a
computational agenda in mind. As we shall see, this leads to a rather
different view of the syllogistic than that defended by certain
latter-day advocates of term-logic.

This paper investigates the following six logics: (i) $\cS$, which
corresponds to the traditional syllogistic; (ii) $\cS^\dagger$, which
extends $\cS$ with negated nouns such as {\sf non-man} or {\sf
non-animal}; (iii) $\cR$, which extends $\cS$ with transitive verbs
such as {\sf kills}; (iv) $\cR^\dagger$, which extends $\cS$ with both
these constructions; (v) $\cR^*$, which extends $\cR$ by allowing
subject noun phrases to contain relative clauses as in the conclusion
of Argument~\eqref{eq:arg1}; and (vi) $\cR^{*\dagger}$, which extends
$\cR^*$ with negated nouns. We derive results on the existence (or
non-existence) of syllogistic proof-systems for these logics, paying
particular regard to the need for the rule of {\em reductio ad
absurdum}.  Specifically, we present sound and complete syllogistic
systems for both $\cS$ and $\cS^\dagger$ that do not employ {\em
reductio} (we call such systems {\em direct syllogistic systems}).  We
then show that, by contrast, there is no sound and complete direct
syllogistic system for the fragment $\cR$; however, we do present a
direct syllogistic system for this fragment that is sound and {\em
refutation-complete} (i.e.~becomes complete if {\em reductio ad
absurdum} is allowed as a single, final step).  We then consider {\em
indirect syllogistic systems}---i.e.~those in which the rule of {\em
reductio ad absurdum} may be employed at any point during the
derivation. We show that---unless \PTIME$=$\NPTIME---there is no
direct syllogistic system that is sound and refutation-complete for
the fragment $\cR^*$; however, we do provide an indirect syllogistic
system that is sound and complete for this fragment.  Finally, we show
that neither $\cR^\dagger$ nor $\cR^{*\dagger}$ has even an indirect
syllogistic system that is sound and complete.  We also obtain results
on the computational complexity of determining validity in the above
fragments, and show how these results are related to the existence of
sound and complete proof-systems (of various kinds) for the fragments
in question. Specifically, we show that the problem of determining the
validity of sequents in any of the fragments $\cS$, $\cS^\dagger$ or
$\cR$ is \NLOGSPACE-complete; the problem of determining the validity
of sequents in the fragment $\cR^*$ is co-\NPTIME-complete; and the
problem of determining the validity of sequents in either of the
fragments $\cR^\dagger$ or $\cR^{*\dagger}$ is \EXPTIME-complete.
Thus, although we are heartened to learn of interest in relational
syllogisms going back to de Morgan, our contribution should not be
read as a historical reconstruction. It is a product of our own time, 
different from de Morgan's in both motivation and outcome.

\section{Preliminaries}
\label{sec:preliminaries}
\subsection{Some syllogistic fragments}
\label{section-fragments}
\paragraph{The fragments $\cS$ and $\cS^\dagger$}
Fix a countably infinite set $\bP$. We may assume $\bP$ to contain
various English common count-nouns such as {\sf man}, {\sf animal}
{\em etc}.  A {\em unary atom} is an element of $\bP$; a {\em unary
  literal} is an expression of either of the forms $p$ or $\bar{p}$,
where $p$ is a unary atom. A unary literal is called {\em positive}
if it is a unary atom; otherwise, {\em negative}. We use the (possibly
decorated) variables $o$, $p$, $q$ to range over unary atoms, and $l$,
$m$, $n$ to range over unary literals. With these conventions, an
$\cS$-{\em formula} is an expression of any of the forms
\begin{equation}
\exists(p,l), \hspace{1cm}
\exists(l,p), \hspace{1cm}
\forall(p,l), \hspace{1cm}
\forall(l,\bar{p}).
\label{eq:syntaxS}
\end{equation}
We provide English glosses for $\cS$-formulas as follows:
\begin{equation}
\mbox{
\begin{minipage}{10cm}
\begin{tabbing}
$\forall(p,q)$ \hspace{0.5cm} \= {\sf Every $p$ is a $q$}
  \hspace{1cm} \= 
$\forall(p,\bar{q})$  \hspace{0.5cm} \= {\sf No $p$ is a $q$} \\
â$\exists(p,q)$ \> {\sf Some $p$ is a $q$} \> 
$\exists(p,\bar{q})$ \> {\sf Some $p$ is not a $q$}.
\end{tabbing}
\end{minipage}
}
\label{eq:gloss}
\end{equation}
Note that $\cS$ contains the formulas $\exists(\bar{p}, q)$ and
$\forall(\bar{p}, \bar{q})$, which are not glossed
in~\eqref{eq:gloss}; however, the semantics given below ensure that
these formulas are logically equivalent to $\exists(q,\bar{p})$ and
$\forall(q,p)$, respectively.  We may regard $\cS$ as the language of
the traditional syllogistic.

The syntax for $\cS$ given above suggests a natural generalization.
An $\cS^\dagger$-{\em formula} is an expression of either of the forms
\begin{equation}
ˆ\exists(l,m), \hspace{1cm}
\forall(l,m).
\label{eq:syntaxS+}
\end{equation}
We provide English glosses for $\cS^\dagger$-formulas in the same way
as for $\cS$, except that we sometimes require {\em negated} subjects:
\begin{center}
\begin{minipage}{11cm}
$\exists(\bar{p},\bar{q})$ \hspace{0.5cm} {\sf Some non-$p$ is not a $q$}
\hspace{.7cm}
$\forall(\bar{p},q)$ \hspace{0.5cm} {\sf Every non-$p$ is a $q$}.
\end{minipage}
\end{center}
We may regard $\cS^\dagger$ the extension of the traditional
syllogistic with noun-level negation.

\paragraph{The fragments $\cR$ and $\cR^\dagger$}
Now fix a countably infinite set $\bR$, disjoint from $\bP$.  We may
assume $\bR$ to contain various English transitive verbs such as {\sf
  kill}, {\sf admire} {\em etc}.  A {\em binary atom} is an element of
$\bR$; a {\em binary literal} is an expression of either of the forms
$r$ or $\bar{r}$, where $r$ is a binary atom.  A binary literal is
called {\em positive} if it is a binary atom; otherwise, {\em
  negative}.  We use the (possibly decorated) variables $r$, $s$ to
range over binary atoms, and $t$ to range over binary literals.  With
these conventions, a {\em c-term} is an expression of any of the forms
\begin{equation}
l, \hspace{1cm} \exists(p,t), \hspace{1cm} \forall(p,t).
\label{eq:syntaxC}
\end{equation}
Thus, all literals are, by definition, c-terms.  We use the variables
$c$, $d$ to range over c-terms.  With this convention, an $\cR$-{\em
  formula} is an expression of any of the forms
\begin{equation}
\exists(p,c), \hspace{1cm}
\exists(c,p), \hspace{1cm}
\forall(p,c), \hspace{1cm}
\forall(c,\bar{p});
\label{eq:syntaxR}
\end{equation}
Thus, all $\cS$-formulas are, by definition, $\cR$-formulas.

We gloss (non-literal) c-terms using complex noun phrases,
as follows:
\begin{align}
& \exists(q,r) \hspace{0.0cm} & & \mbox{\sf thing which $r$s some $q$} 
\label{eq:conceptGlossPos}
& & & \forall(q,r) & & & \mbox{\sf thing which $r$s every $q$}\\
& \exists(q,\bar{r}) & & \mbox{\sf thing which does not $r$ every $q$}
& & & \forall(q,\rbar) & & & \mbox{\sf thing which $r$s no $q$}.
\label{eq:conceptGlossNeg}
\end{align}
And we gloss $\cR$-formulas involving such c-terms accordingly, thus:
\begin{center}
\begin{minipage}{10cm}
\begin{tabbing}
$\forall(p,\exists(q,r))$ \hspace{0.0cm} \= {\sf Every $p$ $r$s some $q$}
  \hspace{0.5cm} \=
$\forall(p,\exists(q,\bar{r}))$  \hspace{0.0cm} \= {\sf No $p$ $r$s every $q$} \\
$\exists(p,\exists(q,r))$ \> {\sf Some $p$ $r$s some $q$} \>
$\exists(p,\exists(q,\bar{r}))$  \> {\sf Some $p$ does
    not $r$ every $q$} \\
$\forall(p,\forall(q,r))$  \> {\sf Every $p$ $r$s every $q$} \> 
$\forall(p,\forall(q,\bar{r}))$  \> {\sf No $p$ $r$s any $q$} \\
$\exists(p,\forall(q,r))$ \> {\sf Some $p$ $r$s every $q$} \>
$\exists(p,\forall(q,\bar{r}))$   \> {\sf Some $p$ $r$s no $q$}.
\end{tabbing}
\end{minipage}
\end{center}
In these glosses, quantifiers in subjects are assumed to have wide
scope; those in objects, narrow scope. Also, the {\sf not} in {\sf
  Some $p$ does not $r$ every $q$} is assumed to scope over the direct
object.  Again, formulas of the forms $\exists(c,p)$ and
$\forall(c,\bar{p})$, which are not glossed here, will turn out to be
equivalent to formulas which are. We may regard $\cR$ as the language
of the {\em relational syllogistic}.

Let us now extend $\cR$ with noun-level negation just as we did with
$\cS$.  An {\em e-term} is an expression of any of the forms
\begin{equation}
l, \hspace{1cm} \exists(l,t), \hspace{1cm} \forall(l,t).
\label{eq:syntaxP}
\end{equation}
We gloss non-literal $e$-terms along the
lines of~\eqref{eq:conceptGlossPos} and~\eqref{eq:conceptGlossNeg}, but
with (possibly) negated direct objects, for example:
\begin{center}
\begin{minipage}{11cm}
$\forall(\bar{q},r)$ \hspace{0.25cm} {\sf thing which $r$s every
    non-$q$}
\hspace{0.5cm}
$\forall(\bar{q},\bar{r})$ \hspace{0.25cm} {\sf thing which $r$s 
    no non-$q$}.
\end{minipage}
\end{center}
(Another way to gloss the latter term would be 
{\sf thing which only $r$s $q$s}.)
Thus, all c-terms are, by definition, e-terms.  We use the variables
$e$, $f$ to range over e-terms. With this convention, an
$\cR^\dagger$-{\em formula} is an expression of any of the forms
\begin{equation}
\exists(l,e), \hspace{1cm}
\exists(e,l), \hspace{1cm}
\forall(l,e), \hspace{1cm}
\forall(e,l),
\label{eq:syntaxR+}
\end{equation}
with the obvious English glosses. Note however that noun-level
negation may be required, both in subjects and direct-objects:
\begin{center}
\begin{minipage}{7cm}
$\exists(\bar{p}, \forall(\bar{q},r))$ \hspace{0.5cm} 
{\sf Some non-$p$ $r$s every non-$q$}.
\end{minipage}
\end{center}
We may regard $\cR^\dagger$ as the extension of the relational
syllogistic with noun-level negation.

\paragraph{The fragments $\cR^*$ and $\cR^{*\dagger}$}
The reader may have noticed that neither $\cR$ nor $\cR^\dagger$ can express
the conclusion of Argument~\eqref{eq:arg1}. We now rectify this
matter by introducing two additional fragments which can.

It will be convenient to extend the `bar'-notation (as in $\bar{p}$
and $\bar{r}$) to all e-terms.  If $l = \bar{p}$ is a negative unary
literal, then we take $\bar{l}$ to denote $p$; and similarly for
binary literals.  If $e$ is an e-term of the form $\forall(l,t)$, we
denote by $\bar{e}$ the corresponding e-term $\exists(l,\bar{t})$; if
$e$ is an e-term of the form $\exists(l,t)$, we denote by $\bar{e}$
the corresponding e-term $\forall(l,\bar{t})$. Thus, for
all e-terms $e$, we have $\bar{\bar{e}} = e$.

Recalling that a c-term $c$ is an expression of any of the forms $l$,
$\exists(p,t)$ or $\forall(p,t)$, we call $c$ {\em positive} if the
literal $l$ or $t$ is positive; otherwise, negative. 
Evidently, if $c$ is a positive c-term, then $\bar{c}$ is a negative
c-term, and vice-versa.  We now define an $\cR^*$-{\em formula} to be
an expression of any of the forms
\begin{equation}
\exists(c^+,d), \hspace{1cm}
\exists(d,c^+), \hspace{1cm}
\forall(c^+,d), \hspace{1cm}
\forall(d,\overline{c^+}),
\label{eq:syntaxR*}
\end{equation}
where $c^+$ ranges over positive c-terms and $d$ over all c-terms
(positive or negative).  These forms may then be glossed by combining
the glosses~\eqref{eq:gloss} and~\eqref{eq:conceptGlossPos} in the
obvious way.  Thus, for example, the conclusion of
Argument~\eqref{eq:arg1} may be formalized in $\cR^*$ as
\begin{center}
\begin{minipage}{12cm}
\begin{tabbing}
$\forall(\exists(\mbox{man},\mbox{kill}),
         \exists(\mbox{animal},\mbox{kill}))$\\
{\sf Everything which kills a man (is a thing which) kills an animal}.
\end{tabbing}
\end{minipage}
\end{center}
It is straightforward to check that, in providing English glosses for
$\cR^*$-formulas, noun-level negation (i.e.~expressions such as {\sf
  non-man}) is not required.  The study of $\cR^*$ is motivated in
part by its close connection to the system described in
McAllester and Givan~\cite{logic:mcA+G92}, and in part by the fact
that, from both a proof-theoretic and complexity-theoretic point of
view, it occupies an intermediate position, in a sense that will be
made precise below.

We come finally to the most general of the fragments studied
here. Recalling that the variables $e$ and $f$ range over e-terms, an
$\cR^{*\dagger}$-{\em formula} is an expression of either of the forms
\begin{equation}
\exists(e,f), \hspace{1cm}
\forall(e,f).
\label{eq:syntaxR*+}
\end{equation}
These formulas receive the obvious glosses, for example:
\begin{center}
\begin{minipage}{12cm}
\begin{tabbing}
$\forall(\exists(\overline{\mbox{animal}},\mbox{kill}),
         \exists(\overline{\mbox{man}},\mbox{kill}))$\\
{\sf Everything which kills a non-animal (is a thing which) kills a non-man}.
\end{tabbing}
\end{minipage}
\end{center}
We may regard $\cR^{*\dagger}$ as the result of extending $\cR^*$ with
noun-level negation.

We denote the set of $\cS$-formulas simply by $\cS$, and similarly for
$\cS^\dagger$, $\cR$, $\cR^\dagger$, $\cR^*$ and $\cR^{*\dagger}$.  Many of the results
established below apply to all of these fragments. Accordingly, we say
that a {\em syllogistic fragment} is any of the fragments $\cS$,
$\cS^\dagger$, $\cR$, $\cR^\dagger$, $\cR^*$ or $\cR^{*\dagger}$, and we use the
variable $\cF$ to range over the syllogistic fragments.
Fig.~\ref{fig:overview} gives an overview of the syllogistic
fragments, and Fig.~\ref{fig:syntax} a quick reference guide to their
syntax.

It will be convenient to extend the bar-notation to formulas.  If
$\phi$ is an $\cR^{*\dagger}$-formula of the form $\forall(e,f)$, we
denote by $\bar{\phi}$ the $\cR^{*\dagger}$-formula
$\exists(e,\bar{f})$; if $\phi$ is an $\cR^{*\dagger}$-formula of the
form $\exists(e,f)$, we denote by $\bar{\phi}$ the
$\cR^{*\dagger}$-formula $\forall(e,\bar{f})$. It is easy to verify
that, $\bar{\bar{\phi}} = \phi$, and that, for any syllogistic
fragment $\cF$, $\phi \in \cF$ implies $\bar{\phi} \in \cF$.
\begin{figure}
\begin{center}
\begin{minipage}{11cm}
\begin{minipage}{3cm}
\begin{picture}(60,70)
\put(40,0){$\cS$}
\put(20,20){$\cR$}
\put(0,40){$\cR^*$}
\put(60,20){$\cS^\dagger$}
\put(40,40){$\cR^\dagger$}
\put(20,60){$\cR^{*\dagger}$}

\put(48,8){\rotatebox{45}{$\subseteq$}}
\put(28,28){\rotatebox{45}{$\subseteq$}}
\put(8,48){\rotatebox{45}{$\subseteq$}}

\put(28,12){\rotatebox{315}{$\supseteq$}}
\put(8,32){\rotatebox{315}{$\supseteq$}}

\put(48,32){\rotatebox{315}{$\supseteq$}}
\put(28,52){\rotatebox{315}{$\supseteq$}}
\end{picture}
\end{minipage}
\begin{minipage}{6cm}
{\small 
\begin{tabular}{|l|l|}
\hline
$\cS$ & syllogistic\\
$\cS^\dagger$ & syllogistic with noun-negation\\
$\cR$ & relational syllogistic\\
$\cR^\dagger$ &  relational syllogistic with noun-negation\\
$\cR^*$ & relational syllogistic with complex subjects\\
$\cR^{*\dagger}$ &  relational syllogistic with complex subjects\\
& and noun-negation\\
\hline
\end{tabular}
}
\end{minipage}
\end{minipage}
\end{center}

\vspace{0.25cm}

\caption{The six syllogistic fragments: $\cS$, $\cS^\dagger$, $\cR$,
  $\cR^\dagger$, $\cR^*$ and $\cR^{*\dagger}$.}
\label{fig:overview} 
\end{figure}
\begin{figure}
\begin{center}
\begin{tabular}{l|l|l}
Expression & Variables & Syntax \\
\hline
unary atom & $o$, $p$, $q$ & \ \\
binary atom & $r$, $s$ & \ \\
unary literal & $l$, $m$, $n$ & $p \  \mid \   \bar{p}$ \\
binary literal & $t$ & $r \  \mid \   \bar{r}$ \\
positive c-term & $b^+$, $c^+$ & 
          $p \  \mid \  \exists(p,r) \  \mid \  \forall(p,r)$\\
c-term & $c$, $d$ & 
          $l \  \mid \  \exists(p,t) \  \mid \  \forall(p,t)$\\
e-term & $e$, $f$ & 
          $l \  \mid \  \exists(l,t) \  \mid \  \forall(l,t)$\\
$\cS$-formula & \ & 
   $\exists(p,l) \  \mid \  \exists(l,p) \  \mid \ 
    \forall(p,l) \  \mid \  \forall(l,\bar{p})$\\
$\cS^\dagger$-formula & \ &  $\exists(l,m) \  \mid \  \forall(l,m)$\\
$\cR$-formula & \ & 
   $\exists(p,c) \  \mid \  \exists(c,p) \  \mid \ 
    \forall(p,c) \  \mid \  \forall(c,\bar{p})$\\
$\cR^\dagger$-formula & \ & 
   $\exists(l,e) \  \mid \  \exists(e,l) \  \mid \ 
    \forall(l,e) \  \mid \  \forall(e,l)$\\
$\cR^*$-formula & \ & 
   $\exists(c^+,d) \  \mid \  \exists(d,c^+) \  \mid \ 
    \forall(c^+,d) \  \mid \  \forall(d,\overline{c^+})$\\
$\cR^{*\dagger}$-formula & \ &  $\exists(e,f) \  \mid \  \forall(e,f)$
\end{tabular}
\end{center}
\caption{Syntax of syllogistic fragments: quick reference guide.}
\label{fig:syntax}
\end{figure}
\paragraph{Semantics}
We now provide semantics for the fragment $\cR^{*\dagger}$ (and hence for
any syllogistic fragment).  A {\em structure} $\fA$ is a triple
$\langle A, \{ p^\fA \}_{p \in \bP}, \{ r^\fA \}_{r \in \bR} \rangle$,
where $A$ is a non-empty set, $p^\fA \subseteq A$, for every $p \in
\bP$, and $r^\fA \subseteq A^2$, for every $r \in \bR$. The set $A$ is
called the {\em domain} of $\fA$. Given a structure $\fA$, we extend
the maps $p \mapsto p^\fA$ and $r \mapsto r^\fA$ to all e-terms by
setting
\begin{align*}
\bar{p}^\fA & = A \setminus p^\fA\\
\bar{r}^\fA & = A^2 \setminus r^\fA\\
\exists(l,t)^\fA & = \{a \in A \mid \mbox{$\langle a,b \rangle \in t^\fA$
                       for some $b \in l^\fA$}\}\\ 
\forall(l,t)^\fA & = \{a \in A \mid \mbox{$\langle a,b \rangle \in t^\fA$ 
                        for all  $b \in l^\fA$}\}.
\end{align*}
We define the truth-relation $\models$ between structures and
$\cR^{*\dagger}$-formulas by declaring $\fA \models \forall(e,f)$ if
and only if $e^\fA \subseteq f^\fA$, and $\fA \models \exists(e,f)$ if
and only if $e^\fA \cap f^\fA \neq \emptyset$.  If $\Theta$ is a set
of formulas, we write $\fA \models \Theta$ if, for all $\theta \in
\Theta$, $\fA \models \theta$. Of course, this defines the
truth-relation for formulas of any syllogistic fragment $\cF$.

A formula $\theta$ is {\em satisfiable} if there exists $\fA$ such
that $\fA \models \theta$; a set of formulas $\Theta$ is {\em
  satisfiable} if there exists $\fA$ such that $\fA \models
\Theta$.  If, for all structures $\fA$, $\fA \models \Theta$ implies
$\fA \models \theta$, we write $\Theta \models \theta$. We take it as
uncontroversial that 
$\Theta\models \theta$ constitutes a rational reconstruction of the
pre-theoretic judgment that a conclusion $\theta$ may be validly
inferred from premises $\Theta$. For example, the valid argument
\begin{equation*}
\mbox{
\begin{minipage}{10cm}
\begin{tabbing}
{\sf Some artist is a beekeeper}\\
{\sf Every artist is a carpenter}\\
\underline{\sf No beekeeper is a dentist \hspace{0.9cm}}\\
{\sf Some carpenter is not a dentist}
\end{tabbing}
\end{minipage}
}
\end{equation*}
corresponds to the valid sequent of $\cS$-formulas:
\begin{multline}
\{
\exists (\mbox{artist},\mbox{beekeeper}),
\forall (\mbox{artist},\mbox{carpenter}),\\
\forall (\mbox{beekeeper},\overline{\mbox{dentist}})\}
\models
\exists (\mbox{carpenter},\overline{\mbox{dentist}}).
\label{eq:seq1}
\end{multline}
Likewise, the valid argument
\begin{equation*}
\mbox{
\begin{minipage}{10cm}
\begin{tabbing}
{\sf Some artist hates some artist}\\
\underline{\sf No beekeeper hates any beekeeper}\\
{\sf Some artist is not a beekeeper}
\end{tabbing}
\end{minipage}
}
\end{equation*}
corresponds to the valid sequent of $\cR$-formulas:
\begin{multline}
\{
\exists(\mbox{artist}, \exists(\mbox{artist},\mbox{hate})),\\
\forall(\mbox{beekeeper},\forall(\mbox{beekeeper}, \overline{\mbox{hate}})) \}
\models
\exists (\mbox{artist},\overline{\mbox{beekeeper}}).
\label{eq:seq2}
\end{multline}
We follow modern practice in taking universal quantification not to
carry existential commitment: thus $\{ \forall(p,q) \} \not \models
\exists(p,q)$.

\paragraph{Absurdity, negation and identifications}
A simple check shows that, for any structure $\fA$ and any e-term $e$,
${\bar{e}}^\fA = A \setminus e^\fA$; hence, $\fA \not \models
\exists(e,\bar{e})$. We refer to a formula of this form as an {\em
  absurdity}.  Even the smallest syllogistic fragment, $\cS$, contains
absurdities, namely, those of the form $\exists(p,\bar{p})$.  Where
the fragment $\cF$ is clear from context, we write $\bot$
indifferently to denote any absurdity in $\cF$.  

For any structure $\fA$ and any $\cR^{*\dagger}$-formula $\phi$, $\fA
\models \phi$ if and only if $\fA \not \models \bar{\phi}$. Also,
$\phi$ and $\bar{\phi}$ belong to the same syllogistic fragments.
Thus, if $\cF$ is a syllogistic fragment and $\phi$ an $\cF$-formula,
we may regard $\bar{\phi}$ as the negation of $\phi$.

For any structure $\fA$ and any e-terms $e$ and $f$, $\fA \models
\exists(e,f)$ if and only if $\fA \models \exists(f,e)$.  Also, these
formulas belong to the same syllogistic fragments.  In the sequel,
therefore, we identify such pairs of formulas, silently transforming
one to the other as necessary.  Similarly, for the pair of formulas
$\forall(e,f)$ and $\forall(\bar{f},\bar{e})$. These identifications
make no essential difference to the results derived below on
syllogistic proof-systems; however, they greatly simplify their
presentation and analysis.

\subsection{Syllogistic rules and \emph{reductio ad absurdum}}
\label{section-refutation-completeness}
Let $\cF$ be a syllogistic fragment. A {\em derivation relation}
$\genProves$ in $\cF$ is a subset of $\bbP(\cF) \times \cF$,
where $\bbP(\cF)$ is the power set of $\cF$.  For
readability, we write $\Theta \genProves \theta$ instead of $\langle
\Theta, \theta \rangle \in \ \genProves$.  We say that $\genProves$ is
{\em sound} if $\Theta \genProves \theta$ implies $\Theta
\models \theta$, and {\em complete} ({\em for} $\cF$)
if $\Theta \models \theta$ implies
$\Theta \genProves \theta$.  A set $\Theta$ of $\cF$-formulas is
\emph{inconsistent} ({\em with respect to} $\genProves$) if $\Theta
\genProves \bot$ for some absurdity $\bot \in \cF$; 
otherwise, \emph{consistent}. A weakening of
completeness called \emph{refutation-completeness} will prove
important in the sequel: $\genProves$ is {\em refutation-complete} if
all unsatisfiable sets $\Theta$ are inconsistent.  Completeness
trivially implies refutation-completeness, but not conversely.  We are
primarily interested in derivation relations induced by two different
sorts of deductive system: direct syllogistic systems and indirect
syllogistic systems. These we now proceed to define.
\paragraph{Direct syllogistic systems}
Let $\cF$ be a syllogistic fragment. We employ the following
terminology.  A {\em syllogistic rule} (sometimes, simply: {\em rule}) in
$\cF$ is a pair $\Theta/\theta$, where $\Theta$ is a finite set
(possibly empty) of $\cF$-formulas, and $\theta$ an $\cF$-formula. We
call $\Theta$ the {\em antecedents} of the rule, and $\theta$ its {\em
  consequent}. The rule $\Theta/\theta$ is {\em sound} if $\Theta
\models \theta$.  We generally display rules in `natural-deduction'
style. For example,
\begin{equation}
\begin{array}{ll}
\infer{\exists (p,o)}
                   {\forall (q,o) & 
             \qquad \exists (p,q)}    
\hspace{3cm} & 
\infer{\exists (p,\bar{o})}
                   {\forall (q,\bar{o}) & 
             \qquad \exists (p,q)}    
\end{array}
\label{eq:dariiFerio}
\end{equation}
where $p$, $q$ and $o$ are unary atoms, are syllogistic rules in
$\cS$, corresponding to the traditional syllogisms {\em Darii} and
{\em Ferio}, respectively.  A {\em substitution} is a function $g =
g_1 \cup g_2$, where $g_1: \bP \rightarrow \bP$ and $g_2: \bR
\rightarrow \bR$. If $\theta$ is an $\cF$-formula, denote by
$g(\theta)$ the $\cF$-formula which results by replacing any atom
(unary or binary) in $\theta$ by its image under $g$, and similarly
for sets of formulas.  An {\em instance} of a syllogistic rule
$\Theta/\theta$ is the syllogistic rule $g(\Theta)/g(\theta)$, where
$g$ is a substitution.

Syllogistic rules which differ only with respect to re-naming of unary
or binary atoms will be informally regarded as identical, because they
have the same instances.  Thus, the letters $p$, $q$ and $o$
in~\eqref{eq:dariiFerio} function, in effect, as {\em variables}
ranging over unary atoms. It is often convenient to display
syllogistic rules using variables ranging over other types of
expressions, understanding that these are just more compact ways of
writing finite collections of syllogistic rules in the official sense.
For example, the two rules~\eqref{eq:dariiFerio} may be more compactly
written
\begin{equation*}
\infer[{\mbox{\small (D1)}}]{\exists (p,l)}
                   {\forall (q,l) & 
             \qquad \exists (p,q)}    
\end{equation*}
where $p$ and $q$ range over unary atoms, but $l$ ranges over unary
{\em literals}.

Fix a syllogistic fragment $\cF$, and let $\X$ be a set of syllogistic
rules in $\cF$. Define $\vdash_\X$ to be the smallest derivation
relation in $\cF$ satisfying:
\begin{enumerate}
\item if $\theta \in \Theta$, then $\Theta \vdash_\X \theta$;
\item if $\{\theta_1, \ldots, \theta_n\}/\theta$ is a rule in $\X$,
  $g$ a substitution, $\Theta = \Theta_1 \cup \cdots \cup \Theta_n$,
  and $\Theta_i \vdash_\X g(\theta_i)$ for all $i$ ($1 \leq i \leq
  n$), then $\Theta \vdash_\X g(\theta)$.
\end{enumerate}
It is simple to show that the derivation relation $\vdash_\X$ is sound
if and only if each rule in $\X$ is sound.

Informally, we imagine chaining together instances of the rules in
$\X$ to construct {\em derivations}, in the obvious way; and we refer
to the resulting proof system as the {\em direct syllogistic system
  defined by} $\X$.  We generally display derivations in
natural-deduction style.  Thus, for example, if $\X$ is any rule set
containing (D1), the derivation \newcommand{\aaaa}{
  \infer[{\mbox{\small (D1)}}] {\exists
    (\mbox{carpenter},\mbox{beekeeper})} {\forall
    (\mbox{artist},\mbox{carpenter}) & \qquad \exists
    (\mbox{artist},\mbox{beekeeper}) } } \newcommand{\bbbb}
{\infer[{\mbox{\small (D1)}}]{\exists
    (\mbox{carpenter},\overline{\mbox{dentist}})} {\forall
    (\mbox{beekeeper},\overline{\mbox{dentist}}) & \aaaa} }
$$\bbbb
$$ 
establishes that
\begin{multline*}
\{
\exists (\mbox{artist},\mbox{beekeeper}),
\forall (\mbox{artist},\mbox{carpenter}),\\
\forall (\mbox{beekeeper},\overline{\mbox{dentist}})\}
\vdash_\X
\exists (\mbox{carpenter},\overline{\mbox{dentist}}),
\end{multline*}
which, given the soundness of (D1), entails the
validity~\eqref{eq:seq1}. Notice, incidentally, that the first
application of (D1) in the above derivation depends on the silent
identification of $\exists(p,q)$ and $\exists(q,p)$. In the sequel, we
reason freely about derivations in order to establish properties of
derivation relations.

Derivation relations defined by direct proof-systems are easily seen
to have polynomial-time complexity.
\begin{lemma}
Let $\cF$ be a syllogistic fragment, and $\X$ a finite set of
syllogistic rules in $\cF$. The problem of determining whether $\Theta
\vdash_\X \theta$, for a given set of $\cF$-formulas $\Theta$ and
$\cF$-formula $\theta$, is in \PTIME.
\label{lma:PTIME}
\end{lemma}
\begin{proof}
Let $\Sigma$ be the set of all atoms (unary or binary) occurring in
$\Theta \cup \set{\theta}$, together with one additional binary atom
$r$.  We first observe that, if there is a derivation of $\theta$ from
$\Theta$ using the rules $\X$, then there is such a derivation
involving only the atoms occurring in $\Sigma$.  For, given any
derivation of $\theta$ from $\Theta$, uniformly replace any unary atom
that does not occur in $\Theta \cup \set{\theta}$ with one that
does. Similarly, uniformly replace any binary atom which does not
occur in $\Theta \cup \set{\theta}$ with one which does (or with $r$
in case $\Theta \cup \set{\theta}$ contains no binary atoms). This
process obviously leaves us with a derivation of $\theta$ from
$\Theta$, using the rules $\X$.

To prove the lemma, let the the total number of symbols occurring in
$\Theta \cup \set{\theta}$ be $n$. Certainly, $|\Sigma| \leq n$.  Let
$\X$ comprise $k_1$ proof-rules, each of which contains at most $k_2$
atoms (unary or binary).  The number of rule instances involving only
atoms in $\Sigma$ is bounded by $p(n) = k_1 n^{k_2}$. Hence, we
need never consider derivations with `depth' greater than $p(n)$. Let
$\Theta_i$ be the set of formulas involving only the atoms in
$\Sigma$, and derivable from $\Theta$ using a derivation of depth $i$
or less ($0 \leq i \leq p(n)$).  Evidently, $|\Theta_i| \leq 
|\Theta| + p(n)$. It is then straightforward to compute the
successive $\Theta_i$ in total time bounded by a polynomial function
of $n$.
\end{proof}

\paragraph{Indirect syllogistic systems}
In addition to syllogistic rules, we consider the traditional rule of
\emph{reductio ad absurdum} (RAA), which allows us to derive a
proposition $\phi$ (in some syllogistic fragment) by first assuming
$\bar{\phi}$ and deriving an absurdity.  Again, we display this rule,
in natural-deduction-style, as
\begin{equation*}
\infer[{\mbox{\small (RAA)}}^i]{\phi}
                               {\infer*{\bot}
                                       {\cdots [\bar{\phi}]^i 
                                         \ \cdots 
                                         [\bar{\phi}]^i \cdots }}.
\end{equation*}
The interpretation is as follows: if an absurdity has been derived
from the set of premises $\Phi$ together with the premise
$\bar{\phi}$, then $\phi$ may be derived from the premises $\Phi$
alone. The premise $\bar{\phi}$ may be used several times (including
zero) in the derivation of the absurdity; moreover, $\Phi$ is allowed
to contain $\bar{\phi}$.  We say that the (zero or more) bracketed
instances of the premise $\bar{\phi}$ have been {\em discharged} by
application of (RAA). The numerical superscript $i$ simply allows us
to keep track of which application of (RAA) was responsible for the
discharge of which (instances of a) premise.

Fix a syllogistic fragment $\cF$, and let $\X$ be a set of syllogistic
rules in $\cF$. Define $\Vdash_\X$ to be the smallest derivation
relation in $\cF$ satisfying:
\begin{enumerate}
\item if $\theta \in \Theta$, then $\Theta \Vdash_\X \theta$;
\item if $\{\theta_1, \ldots, \theta_n\}/\theta$ is a rule in $\X$,
  $g$ a substitution, $\Theta = \Theta_1 \cup \cdots \cup \Theta_n$,
  and $\Theta_i \Vdash_\X g(\theta_i)$ for all $i$ ($1 \leq i \leq
  n$), then $\Theta \Vdash_\X g(\theta)$;
\item if $\Theta \cup \{\bar{\theta}\} \Vdash_\X \bot$, then
   $\Theta \Vdash_\X \theta$.
\end{enumerate}
It is again simple to show that $\Vdash_\X$ is sound if and only if
each rule in $\X$ is sound.

Informally, we imagine chaining together instances of the rules in
$\X$ and the rule (RAA) to form {\em indirect derivations} in the
obvious way; and we refer to the resulting proof system as the {\em
  indirect syllogistic system defined by} $\X$. We generally display
indirect derivations in natural-deduction style, as the following
example illustrates. First, consider the (evidently sound) syllogistic
rules:
\begin{equation*}
\infer[{\mbox{\small (D3)}}]{\exists (p,\bar{q})}
                   {\forall (q,\bar{c}) & 
             \qquad \exists (p,c)} 
\hspace{2cm}
\infer[{\mbox{\small ($\forall\forall$)}.}]{\forall (p,\exists(q, t))}
                     {\forall (p, \forall (o, t)) & 
                             \qquad 
                             \exists (q, o)}
\end{equation*}
Here, as usual, $o$, $p$ and $q$ range over unary atoms, $t$ over binary
literals and $c$ over c-terms. In addition, we generalize the rule
(D1) given above so that the variable $l$ (ranging over literals) may
be replaced by the variable $c$ (ranging over c-terms). Then we have the
following indirect derivation.  {\small
\begin{equation*}
\infer[{\mbox{\tiny (RAA)$^1$}.}]{\forall(\mbox{art},\overline{\mbox{bkpr}})}
     {\infer[{\mbox{\tiny (D3)}}]
       {\exists(\mbox{bkpr},\overline{\mbox{bkpr}})}
       {\forall(\mbox{bkpr},\overline{\exists(\mbox{bkpr},\mbox{hate})})
            &
         {\infer[{\mbox{\tiny (D1)}}]
           {\exists(\mbox{bkpr},\exists(\mbox{bkpr},\mbox{hate}))}
           {\infer[{\mbox{\tiny ($\forall\forall$)}}]
              {\forall(\mbox{art},\exists(\mbox{bkpr},\mbox{hate}))}
              {\forall(\mbox{art},\forall(\mbox{art},\mbox{hate}))
               &
               [\exists(\mbox{art},\mbox{bkpr})]^1
              }
            &
            [\exists(\mbox{art},\mbox{bkpr})]^1
           }
         }
       }
     }
\end{equation*}
}

\noindent 
Here, two instances of the premise $\exists(\mbox{artist},\mbox{beekeeper})$
have been discharged by the final application of (RAA). This
derivation establishes that, if $\X$ is any set of rules containing (D3),
($\forall\forall$) and (D1) (generalized as indicated), then
\begin{multline*}
\{\forall(\mbox{artist},\forall(\mbox{artist},\mbox{hate})), \\
  \forall(\mbox{beekeeper},\overline{\exists(\mbox{beekeeper},\mbox{hate})})\}
\Vdash_\X
\forall(\mbox{artist},\overline{\mbox{beekeeper}}).
\end{multline*}
Bearing in mind that $\overline{\exists(\mbox{beekeeper},\mbox{hate})}
= \forall(\mbox{beekeeper},\overline{\mbox{hate}})$, and given the
soundness of the syllogistic rules employed, this entails the
validity~\eqref{eq:seq2}.

It is important to realize that (RAA) is not itself a syllogistic
rule.  In particular, an application of (RAA) in general {\em
  decreases} the set of premises of the derivations in which it
features. As we shall see below, the special status of (RAA) is
essential: indirect syllogistic systems are in general more powerful
than direct syllogistic systems.

If $\vdash_\X$ is refutation-complete, then $\Vdash_\X$ is complete.
For suppose $\Theta \models \theta$. Then $\Theta \cup \{
\bar{\theta}\}$ is unsatisfiable; hence, by the
refutation-completeness of $\vdash_\X$, $\Theta \cup \{ \bar{\theta}\}
\vdash_{\X} \bot$; hence, using a single, final application of (RAA),
$\Theta \Vdash_{\X} \theta$. We stress, however, that (RAA) is not in
general restricted to the final step in an $\Vdash_{\X}$-derivation;
rather, it may be employed at any point, and any number of times.  In
particular, the reasoning of Lemma~\ref{lma:PTIME} fails: in
Section~\ref{sec:noProofR*} we present a set of syllogistic rules
$\R^*$ such that $\Vdash_{\R^*}$ is co-\NPTIME-hard.

We conclude this discussion by showing that indirect proof-systems
yield a version of Lindenbaum's Lemma. Let $\cF$ be any syllogistic
fragment. We say that a set of $\cF$-formulas $\Theta$ is $\cF$-{\em
  complete} if, for every $\cF$-formula $\theta$, either $\theta \in
\Theta$ or $\bar{\theta} \in \Theta$.  If the fragment $\cF$ is clear
from context, we say {\em complete} instead of $\cF$-complete; notice
however that $\cF$-completeness of a set of formulas has nothing to do
with the completeness of a derivation relation.
\begin{lemma} 
Let $\cF$ be a syllogistic fragment, $\Theta$ a set of $\cF$-formulas
and $\X$ a set of syllogistic rules in $\cF$.  If $\Theta$ is
consistent with respect to $\Vdash_\X$, then there exists a complete
set of $\cF$-formulas $\Delta \supseteq \Theta$ consistent with
respect to $\Vdash_\X$.
\label{lma:lindenbaum}
\end{lemma}
\begin{proof}
Let $(\phi_n)_{n\in\N}$ enumerate the formulas of $\cF$.  We define a
sequence of consistent sets $\Delta_n$ as follows.  Let $\Delta_0 =
\Theta$.  For $n \geq0$, let
\begin{equation*}
\Delta_{n+1} = 
\begin{cases}
\Delta_n  \cup \{\phi_{n}\} \mbox{ if } \Delta_n  \cup \{\phi_{n}\}
                           \mbox{ is consistent} \\
\Delta_n  \cup \{\bar{\phi}_{n}\}
 \mbox{ otherwise. }
\end{cases}
\end{equation*}
and let $\Delta = \bigcup_n \Delta_n$. We show by induction that each
$\Delta_n$ is consistent.  For suppose $\Delta_n$ is consistent, but
$\Delta_{n+1}$ inconsistent.
Then $\Delta_n\cup\set{\phi_n}\Vdash_{\X} \bot$, so that $\Delta_n
\Vdash_{\X} \bar{\phi}_n$ using (RAA).  In addition,
$\Delta_n\cup\set{\bar{\phi}_n}\Vdash_{\X} \bot'$ (where $\bot'$ is
some absurdity).  We take a derivation establishing the latter
assertion and replace each (undischarged) premise
$\bar{\phi}$ with a derivation of $\bar{\phi}$ from $\Delta_n$.  This
shows $\Delta_n$ to be inconsistent, a contradiction. Thus, $\Delta_n$
is consistent, for all $n$.

Since derivations are finite, $\Delta$ is consistent. Obviously,
$\Delta$ is $\cF$-complete.
\end{proof}

\section{$\cS$ and $\cS^\dagger$: direct systems}
\label{sec:proofSS+}

In the previous section, we defined our general notions, including the
fragments $\cS$ and $\cS^\dagger$.  These are the simplest in our paper,
coming without verbs.  This section provides sound and complete
syllogistic logics for them.  Let $\S$ be the following set of rules,
where $p$ and $q$ range over unary atoms, $l$ over unary literals, and
and $\phi$, $\psi$ over $\cS$-formulas.
\begin{equation*}
\begin{array}{ccc}
\infer[{\mbox{\small (D1)}}]{\exists (p,l)}
                   {\forall (q,l) & 
             \qquad \exists (p,q)}    
& \hspace{.5 cm}  
\infer[{\mbox{\small (B)}}]{\forall (p,l)}
                   {\forall (p,q) & 
             \qquad \forall (q,l)}
  & \hspace{.5 cm}  \infer[{\mbox{\small (A)}}]{\forall (p,l)}
                   {\forall (p, \bar{p})}  
      \\
\\
\infer[{\mbox{\small (D2)}}]{\exists (q,l)}
                   {\exists (p,l) & 
             \qquad \forall (p,q)}
& \hspace{.5 cm}  
\infer[{\mbox{\small (X)}}]{\phi}{\psi & \bar{\psi}} 
& \hspace{.5cm}  
\ \\
\ \\
\infer[{\mbox{\small (D3)}}]{\exists (p,\bar{q})}
                   {\forall (q,\bar{l}) & 
             \qquad \exists (p,l)}    
& \hspace{.5cm}  
\infer[{\mbox{\small (T)}}]{\forall (p,p)}{}
\qquad
\infer[{\mbox{\small (I)}}]{\exists (p,p)} {\exists (p,l)}  
\end{array}
\end{equation*}
Rules (D1) and (D3) we have met already (the latter in a rather more
general form); Rules (D2) and (B) are new, but evidently versions of
classical syllogisms.  Rule (X) is the classical rule of \emph{ex
  falso quodlibet}: from a contradiction, the reasoner may have
anything he pleases [quodlibet]. This rule is not to be confused with
(RAA): (X) is a syllogistic rule, in the technical sense of this
paper; (RAA) is not. Rules (A), (T) and (I) have no classical
counterparts. Rule (A) stems from the fact that if all $p$ are
non-$p$, then there are no $p$ whatsoever; vacuously, then, all $p$
are $l$.  To see that (T) is needed, note that without it there would
be no way to derive $ \forall(p,p)$ from the empty set of premises.
Rule (I) is self-explanatory. We remark that Rule (D1) is actually
redundant in this context.  For consider any instance $\{\forall
(q',l'), \exists (p',q')\}/\exists (p',l')$ of (D1). If $l'= o$ is a
positive literal, then, using the identification $\exists(e,f) =
\exists(f,e)$, we may re-write this instance as $\{\forall (q',o),
\exists (q',p')\}/\exists (o,p')$, which is an instance of (D2).  On
the other hand, if $l' = \bar{o}$, then, using the identification
$\forall(e,f) = \forall(\bar{f},\bar{e})$, we may re-write the
instance as $\{\forall (o, \bar{q}'), \exists (p',q')\}/\exists
(p',\bar{o})$, which is an instance of (D3). We retain Rule (D1) for
ease of use, and to make the relationship to the system $\R$
introduced in Section~\ref{subsec:R} more transparent.

Turning now to the fragment $\cS^\dagger$, let $\S^\dagger$ comprise the following
syllogistic rules, where $l$, $m$ and $n$ range over unary literals,
and $\phi$, $\psi$ over $\cS^\dagger$-formulas.
\begin{equation*}
\begin{array}{ccc}
\infer[{\mbox{\small (D)}}]{\exists (m,n)}
                   {\exists (l,n) & 
             \qquad \forall (l,m)}    
& \hspace{.5 cm}  
\infer[{\mbox{\small (B)}}]{\forall (l,n)}
                   {\forall (l,m) & 
             \qquad \forall (m,n)}
  & \hspace{.5 cm}  \infer[{\mbox{\small (A)}}]{\forall (l,m)}
                   {\forall (l, \bar{l})}  
      \\
\\
\infer[{\mbox{\small (T)}}]{\forall (l,l)}{}
\qquad   \infer[{\mbox{\small (I)}}]{\exists (l,l)} {\exists (l,m)}  
& \hspace{.5 cm}  
  \infer[{\mbox{\small (X)}}]{\phi}{\psi & \bar{\psi}} 
& \hspace{.5cm}  
\infer[{\mbox{\small (N)}}]{\exists (l, l)}
                   {\forall (\bar{l}, l)} 
\end{array}
\end{equation*}

\noindent
The rules in $\S^\dagger$ are the natural  generalizations of those in
$\S$ to reflect the more liberal syntax of $\cS^\dagger$, but with two
small changes: first, rules (D1)--(D3) have merged into a single rule
(D); second, an extra rule (N) has been added.  The soundness of (N)
is due to the requirement that our structures have non-empty domains.

That $\S^\dagger$ really is an extension of $\S$ is given by the following
lemma.
\begin{lemma}
Any instance of a rule in $\S^\dagger$ which involves only formulas in the
fragment $\cS$ is an instance of a rule of $\S$, and conversely.
\label{lma:S+S} 
\end{lemma}
\begin{proof}
We prove the first statement by considering the rules of $\S^\dagger$
one at a time.  We illustrate here with Rule (D). Suppose that all
formulas in some instance $\{\exists(l,n), \forall (l,m)\}/\exists
(m,n)$ of this rule are in $\cS$.  If $m$ is positive, then so is $l$
(since $\forall(l,m) \in \cS$).  In this case, our instance matches
Rule (D2) in the system $\S$.  On the other hand, if $m$ is negative,
then $n$ is positive (since $\exists(m,n) \in \cS$).  We then use the
identifications $ \forall (l,m) = \forall (\bar{m},\bar{l})$,
$\exists(l,n)= \exists(n,l)$, and $\exists(m,n)= \exists(n,m)$, and
re-write the rule-instance as
$\{\exists(n,l),\forall(\bar{m},\bar{l})\}/\exists (n,m)$, which
(bearing in mind that $n$ and $\bar{m}$ are unary atoms) matches Rule
(D3) in the system $\S$. The other rules are dealt with
similarly. Note in particular that all instances of Rule (N) involve a
formula lying outside the fragment $\cS$, so that the statement holds
trivially for this rule.  The second statement of the lemma is
completely routine.
\end{proof}
\noindent
In view of Lemma~\ref{lma:S+S}, we have taken the liberty of using the
same names---(B), (A), (T), (I) and (X)---for corresponding rules in
$\S$ and $\S^\dagger$.

It is obvious that $\vdash_\S$ and $\vdash_{\S^\dagger}$ are sound. We
now prove their completeness (for the fragments for $\cS$ and
$\cS^\dagger$, respectively).  It is convenient to prove the
completeness of $\vdash_{\S^\dagger}$, and then to derive the result
for $\vdash_\S$ as a special case.

\paragraph{Starting on the proof}
In the remainder of this section, then, a {\em formula} is an
$\cS^\dagger$-formula unless otherwise stated.  A {\em universal} formula is
one of the form $\forall(l,m)$; an {\em existential} formula is one of
the form $\exists(l,m)$. If $\fA$ and $\fB$ are structures with
disjoint domains $A$ and $B$, respectively, denote by $\fA \cup \fB$
the structure with domain $A \cup B$ and interpretations $p^{\fA \cup
\fB} = p^\fA \cup p^\fB$ for any unary atom $p$. (Note that $\cS^\dagger$
features only unary atoms.)
\begin{lemma}
Suppose $\Phi \cup \Psi \models \theta$, where $\Phi$ is a set of
universal formulas, $\Psi$ a set of existential formulas, and $\theta$
a formula.
\begin{enumerate}
\item If $\Phi \cup \Psi$ is satisfiable and $\theta$ is universal, then
$\Phi \models \theta$. 
\item If $\Psi \neq \emptyset$ and $\theta$ is existential, then
there exists $\psi \in \Psi$ such that $\Phi \cup \{ \psi \} \models \theta$. 
\item If $\Psi = \emptyset$ and $\theta = \exists(l,m)$ is existential, then
$\Phi \models \forall(\bar{l},l)$ and $\Phi \models \forall(\mbar,m)$.
\end{enumerate}
\label{lma:1}
\end{lemma}
\begin{proof}
In each case, assume the contrary.
\begin{enumerate}
\item 
There exist structures $\fA \models \Phi
\cup \Psi$ and $\fB \models \Phi \cup {\bar{\theta}}$. 
We may assume that $A\cap B = \emptyset$.
But then $\fA
\cup \fB \models \Phi \cup \Psi \cup \{ \bar{\theta} \}$, a
contradiction.  
\item For every $\psi \in \Psi$, there exists a structure
$\fA_\psi$ such that $\fA_\psi \models \Phi \cup \{ \psi,
\bar{\theta} \}$.   Again, we assume that the domains are pairwise disjoint.
But then $\bigcup_{\psi \in \Psi} \fA_\psi \models
\Phi \cup \Psi \cup \{ \bar{\theta} \}$, a contradiction.
\item If $\Phi \not \models \forall(\bar{l},l)$, there
exists a structure $\fA$ such that $\fA \models \Phi \cup
\{\exists(\bar{l},\bar{l})\}$. Choose $a \in {\bar{l}}^\fA$, and let $\fB$ be
the
structure obtained by restricting $\fA$ to the singleton domain $\{ a
\}$. Then $\fB \models \Phi \cup \{ \bar{\theta} \}$, a contradiction. A
similar argument applies if $\Phi \not \models \forall(\bar{m},m)$. 
\end{enumerate}
\end{proof}

We write $\Phi \vdash_{\rm BTA} \phi$ if there is a derivation in
$\vdash_{\S^\dagger}$ of $\phi$ from $\Phi$ employing only the rules (B), (T)
and (A). We call a set $V$ of literals {\em consistent} if $l \in V$
implies $\bar{l} \not \in V$ (otherwise {\em inconsistent}); we call
$V$ {\em complete} if $l \not \in V$ implies $\bar{l} \in V$.  Let
$\Phi$ be a set of universal formulas and $V$ a set of
  literals.  Define $S^\Phi$ to be the set
of literals:
\begin{equation*}
\{ m : \mbox{$\Phi \vdash_{\rm BTA} \forall(l,m)$ for some $l \in V$} \}.
\end{equation*}
We say that $V$ is $\Phi$-{\em closed} if $V^\Phi \subseteq V$.  

\begin{lemma}
Let $\Phi$ be a set of universal formulas and $V$ a set of
literals.  Then $V^\Phi$ is $\Phi$-closed, and $V \subseteq V^\Phi$.
\label{lma:2}
\end{lemma}
\begin{proof}
Almost immediate, by rules (B) and (T).
\end{proof}
Evidently, the union of any collection of $\Phi$-closed sets of
literals is $\Phi$-closed.  
\begin{lemma}
Let $\Phi$ be a set of universal formulas and $V$ a set of
literals.  If $V^\Phi$ is inconsistent, then there exist literals $l,
l' \in V$ such that $\Phi \vdash_{\rm BTA} \forall(l,\bar{l}')$.
\label{lma:2.5}
\end{lemma}
\begin{proof}
If $m, \mbar \in V^\Phi$, pick $l, l' \in V$ with $\Phi \vdash_{\rm
  BTA} \forall(l,m)$ and $\Phi \vdash_{\rm BTA} \forall(l',\mbar)$.
Re-writing $\forall(l',\mbar)$ as $\forall(m,\bar{l}')$, we have a
derivation from $\Phi$
\begin{equation*}
\infer[(\mathrm B),]{\forall(l,\bar{l}')}
    {\infer*{\forall(l,m)}{} \qquad & \infer*{\forall(m,\bar{l}')}{}}
\end{equation*}
as claimed.
\end{proof}
\begin{lemma}  
Let $\Phi$ be a set of universal formulas.  Any
non-empty, $\Phi$-closed, consistent set of literals has a
$\Phi$-closed, consistent, complete extension.
\label{lma:3}
\end{lemma}
\begin{proof}
Let $V_0$ be a non-empty, $\Phi$-closed, consistent set of literals.
Enumerate the set of all literals   as $l_1, l_2, \ldots $.  For all $i
\geq 0$, define
\begin{equation*}
V_{i+1} = 
\begin{cases}
V_i \mbox{ if } l_{i+1} \in V_i \\
\big(V_i \cup \{\bar{l}_{i+1}\}\big)^\Phi \mbox{ otherwise. }
\end{cases}
\end{equation*}
and define $V = \bigcup_i V_i$. Thus, $V$ is $\Phi$-closed, complete
and includes $V_0$.  We remark that, since $V_0 \subseteq V_i$, $V_i$
is non-empty for all $i \geq 0$.  To show that $V$ is consistent,
suppose otherwise. Let $i$ be the least natural number such that
$V_{i+1}$ is inconsistent. Hence, $V_{i+1} = \big(V_i \cup
\{\bar{l}_{i+1}\}\big)^\Phi$, and so by Lemma~\ref{lma:2.5} there
exist literals $l, l' \in V_i \cup \{\bar{l}_{i+1}\}$ such that $\Phi
\vdash_{\rm BTA} \forall(l,\bar{l}')$.  Since $V_i$ is $\Phi$-closed
and consistent, $l$ and $l'$ cannot both be in $V_i$, and so we may
assume without loss of generality that $l' = \bar{l}_{i+1}$, whence
$\Phi \vdash_{\rm BTA} \forall(l, l_{i+1})$.  Now, either $l \in V_i$
or $l = \bar{l}_{i+1}$. In the former case, we have $l_{i+1} \in V_i$,
because $V_i$ is $\Phi$-closed; in the latter, let $l''$ be any
literal in $V_i$ (which we know to be non-empty).  Then the inference
\begin{equation*}
\infer[\mbox{\rm (A)}]{\forall(l'',l_{i+1})}
           {\forall(\bar{l}_{i+1},l_{i+1})}
\end{equation*}
guarantees that, again, $l_{i+1} \in V_i$. Either way, $V_{i+1} =
V_i$, a contradiction.
\end{proof}
\begin{theorem}
The derivation relation $\vdash_{\S^\dagger}$ is sound and complete
for $\cS^\dagger$.
\label{theo:S+completeness}
\end{theorem}
\begin{proof}
Soundness is routine.  For the converse, let $\Theta$ be a set of
$\cS^\dagger$-formulas and $\theta$ a $\cS^\dagger$-formula, and suppose $\Theta
\models \theta$. By the compactness theorem for first-order logic, we
may safely assume that $\Theta$ is finite. Suppose for the moment that
$\Theta$ is satisfiable, and write $\Theta = \Phi \cup \Psi$, where
$\Phi$ is a set of universal formulas and $\Psi$ a set of existential
formulas. We consider three cases: (1) $\theta$ is universal; (2)
$\theta$ is existential and $\Psi$ is non-empty; and (3) $\theta$ is
existential and $\Psi$ is empty.

In the remainder of this proof, we simplify our notation to write
$\vdash$ for $\vdash_{\S^\dagger}$.

\vspace{0.25cm}

\noindent Case (1): Write $\theta = \forall (l_0,m_0)$. By
Lemma~\ref{lma:1}, Part 1, $\Phi \models \theta$. Let
\begin{equation}
V_0 = \{ l_0, \mbar_0 \}^\Phi.
\end{equation}
Then $V_0$ is non-empty and $\Phi$-closed, by Lemma~\ref{lma:2}.  We
claim that $V_0$ is inconsistent. For suppose otherwise.  By
Lemma~\ref{lma:3}, let $V$ be a consistent complete extension of $V_0$,
and define $\fA$ to be the structure with singleton domain $\{a \}$
given by
\begin{equation*}
p^\fA = 
\begin{cases}
\{ a \} \mbox{ if $p \in V$}\\
\emptyset \mbox{ otherwise,}
\end{cases}
\end{equation*}
for every atom $p$.  It is easily seen that $\fA \models \Phi \cup
\bar{\theta}$, a contradiction. So by Lemma~\ref{lma:2.5}, there exist
literals $l, l' \in \{ l_0, \mbar_0 \}$ such that
\begin{equation}
\Phi \vdash_{\rm BTA} \forall(l,\bar{l}').
\label{eq1}
\end{equation}
By exchanging $l$ and $l'$ if necessary, we have two sub-cases: (i)
$l= l_0$ and $l' = \mbar_0$; (ii) $l = l' \in \{l_0,\mbar_0\}$.  
In sub-case (i), \eqref{eq1} simply asserts that $\Phi \vdash \theta$.
In sub-case (ii), we have one of the derivations
from $\Phi$
\begin{equation*}
\infer[({\mathrm A})]{\forall(\mbar_0,\bar{l}_0)}
      {\infer*{\forall(l_0,\bar{l}_0)}{}}
\hspace{2cm}
\infer[({\mathrm A}),]{\forall(l_0,m_0)}
      {\infer*{\forall(\mbar_0,m_0)}{}}
\end{equation*}
and so $\Phi \vdash \theta$. 

\vspace{0.25cm}

\noindent Case (2): Write $\theta = \exists (l,m)$. By
Lemma~\ref{lma:1}, Part~2, there exists $\psi = \exists (l_0,m_0) \in
\Psi$ such that $\Phi \cup \{ \psi \} \models \theta$. 
Set
\[
V_0 = \{ l_0, m_0, \bar{l} \}^\Phi.
\]
The set $V_0$ must be inconsistent. For otherwise, we can easily
construct, using a parallel argument to that employed in Case (1), a
structure $\fA$ such that $\fA \models \Phi \cup \{ \psi, \bar{\theta}
\}$, contradicting the fact that $\Phi \cup \{ \psi \} \models
\theta$. Hence, there exist literals $l_1, l_2 \in \{ l_0, m_0,
\bar{l} \}$ such that $\Phi \vdash_{\rm BTA}
\forall(l_1,\bar{l}_2)$. If $l_1$ and $l_2$ are both in $\{ l_0, m_0
\}$, then $\Theta$ is unsatisfiable, contrary to
hypothesis.  So assume, without loss of generality, that $l_2 =
\bar{l}$. Thus, $\Phi \vdash_{\rm BTA} \forall(l_1,l)$, and we have the
following possibilities: (i) $l_1 = l_0$; (ii) $l_1 = m_0$; (iii) $l_1
= \bar{l}$. Possibility (i) yields $\Phi
\vdash \forall(l_0,l)$. Possibility (ii) yields $\Phi \vdash
\forall(m_0,l)$.  Possibility (iii) also yields $\Phi
\vdash \forall(l_0,l)$, via the derivation
\begin{equation*}
\infer[(\mathrm A).]{\forall(l_0,l)}
                   {\infer*{\forall(\bar{l},l)}{}}
\end{equation*}
In other words, we have proved:
\begin{equation}
\mbox{either } \Phi \vdash \forall(l_0,l) \mbox{ or } \Phi \vdash \forall(m_0,l).
\label{eq:possibilities1}
\end{equation}
Replacing $l$ by $m$ in the above argument yields, in exactly the
same way:  
\begin{equation}
\mbox{either } \Phi \vdash \forall(l_0,m) \mbox{ or } \Phi \vdash \forall(m_0,m).
\label{eq:possibilities2}
\end{equation}
Considering~\eqref{eq:possibilities1}, we may assume, by transposing
$l_0$ and $m_0$ if necessary, that $\Phi \vdash \forall(l_0,l)$. This
leaves us with the two possibilities in~\eqref{eq:possibilities2}.  If
$\Phi \vdash \forall(m_0,m)$, we have
\begin{equation*}
\infer[(\mathrm D);]{\exists(l,m)}
                   {\infer[(\mathrm D)]{\exists(l,m_0)}  
                                       {\exists(l_0,m_0) 
                                        \qquad &
                                        \infer*{\forall(l_0,l)}
                                               {}} 
                    \qquad &
                    \infer*{\forall(m_0,m)}{}}
\end{equation*}
if, on the other hand, 
$\Phi \vdash \forall(l_0,m)$, we have
\begin{equation}
\infer[(\mathrm D).]{\exists(l,m)}
                   {\infer[(\mathrm D)]{\exists(l,l_0)}  
                                       {\infer[(\mathrm I)]{\exists(l_0,l_0)} 
                                                           {\exists(l_0,m_0)}
                                        \qquad &
                                        \infer*{\forall(l_0,l)}
                                               {}} 
                    \qquad &
                    \infer*{\forall(l_0,m)}{}}
\label{eq:d2}
\end{equation}
Either way, $\Theta \vdash \theta$, as required.

\vspace{0.25cm}

\noindent Case (3): Write $\theta = \exists (l,m)$. Since $\Theta =
\Phi \models \theta$, by Lemma~\ref{lma:1}, Part 3, $\Phi \models
\forall(\bar{l},l)$ and $\Phi \models \forall(\bar{m},m)$. By Case
(1), $\Phi \vdash \forall(\bar{l},l)$ and $\Phi \vdash
\forall(\bar{m},m)$. Therefore, we have the derivation from $\Phi$
\begin{equation*}
\infer[(\mathrm D),]{\exists(l,m)}
                   {\infer[(\mathrm N)]{\exists(l,l)}  
                                       {\infer*{\forall(\bar{l},l)}{}} 
                    \qquad &
                    \infer[(\mathrm A)]{\forall(l,m)}
                                       {\infer*{\forall(\bar{m},m)}{}}}\
\end{equation*}
and $\Theta \vdash \theta$.

\vspace{0.25cm}

\noindent We have now shown that, for $\Theta$ satisfiable, 
$\Theta \models \theta$ implies $\Theta \vdash \theta$. It remains
only to consider the case where $\Theta$ is unsatisfiable. If so,
let $\Theta' \cup \{\theta'\}$ be a minimal unsatisfiable subset of
$\Theta$. (Remember, we are allowed to assume that $\Theta$ is finite.)
Hence, $\Theta'$ is satisfiable, with
$\Theta' \models \bar{\theta}'$. By the previous argument, 
$\Theta' \vdash \bar{\theta}'$. Thus, we have the derivation from 
$\Theta' \cup \{\theta'\}$
\begin{equation*}
\infer[(\mathrm X),]{\theta}
                   {\infer*{\bar{\theta}'}{}
                    \qquad &
                    \theta'}
\end{equation*}
and $\Theta \vdash \theta$. 
\end{proof}

\begin{theorem}  
The derivation relation $\vdash_\S$ is sound and complete for $\cS$.
\label{theo:Scompleteness-redux}
\end{theorem}
\begin{proof}
Soundness is obvious. Suppose that $\Theta$ is a satisfiable set of
$\cS$-formulas and $\theta$ an $\cS$-formula such that $\Theta \models
\theta$.  By Theorem~\ref{theo:S+completeness}, $\Theta \vdash_{\S^\dagger}
\theta$.  We must show that $\Theta \vdash_{\S} \theta$.  Since
derivations (in $\vdash_{\S^\dagger}$) are finite, we may assume without
loss of generality that $\Theta$ is finite. For the moment, let us
further assume that $\Theta$ is satisfiable.

Let $p$, $q$ be any unary atoms.  We claim that $\Theta \not \models
\forall(\bar{p},q)$.  For otherwise, by Lemma~\ref{lma:1} Part~1,
there exists a subset $\Phi \subseteq \Theta$ of {\em universal}
formulas such that $\Phi \models \forall(\bar{p},q)$. Now let $\fA$ be
any structure such that $o^\fA= \emptyset$ for every unary atom
$o$. Since $\Phi \subseteq \cS$, $\fA \models \Phi$; but $\fA \not
\models \forall(\bar{p},q)$, a contradiction. Furthermore, since
$\vdash_{\S^\dagger}$ is sound, it follows of course that $\Theta \not
\vdash_{\S^\dagger} \forall(\bar{p},q)$.

Consider any derivation in $\vdash_{\S^\dagger}$ with premises $\Theta$.  We
claim that, if this derivation contains any formula of the form
$\exists(\bar{p},\bar{q})$, then the final conclusion of this
derivation is also of that form. For, since $\Theta$ is, by
assumption, satisfiable, (X) cannot be used in the derivation; and the
only other rules in $\S^\dagger$ with any premise of the form
$\exists(\bar{p},\bar{q})$ are (I) and (D), thus:
\begin{equation*}
\infer[{\mbox{\small (I)}}]{\exists (\bar{p},\bar{p})} {\exists (\bar{p},\bar{q})}
\hspace{1cm}
\infer[{\mbox{\small (D)}}.]{\exists (\bar{q},l)}
                   {\exists (\bar{p},\bar{q}) & 
             \qquad \forall (\bar{p},l)}    
\end{equation*}
By the observation of the previous paragraph, the literal $l$ in this
instance of (D) must be negative. In either case, the consequent of
the rule is of the form $\exists(\bar{p},\bar{q})$, as claimed.

Now take any derivation of $\theta$ from $\Theta$ in
$\vdash_{\S^\dagger}$. (We know that one exists.)  Since $\theta \in
\cS$ is definitely not of the form $\exists(\bar{p},\bar{q})$, it
follows from the previous two paragraphs that this derivation cannot
involve any formula of either of the forms $\forall(\bar{p},q)$ or
$\exists(\bar{p},\bar{q})$. That is, all the formulas are in $\cS$. By
Lemma~\ref{lma:S+S}, $\Theta \vdash_{\S} \theta$.

This proves the theorem in the case where the (finite) set of
$\cS$-formulas $\Theta$ is satisfiable.  If $\Theta$ is unsatisfiable,
let $\Theta' \cup \{\theta'\}$ be a minimal unsatisfiable subset of
$\Theta$. Hence, $\Theta'$ is satisfiable, with $\Theta' \models
\bar{\theta}'$.  By the result just established, $\Theta' \vdash_{\S}
\bar{\theta}'$.  By a single application of (X), $\Theta \vdash_{\S}
\theta$.
\end{proof}

\section{$\cR$: a refutation complete system}
\label{sec:proofR}
We turn from $\cS$ to $\cR$.  We exhibit a set $\R$ of syllogistic
rules in $\cR$, and prove that $\vdash_\R$ is a sound and {\em
  refutation}-complete derivation relation for $\cR$; we also prove
that there is no finite set of rules $\X$ in $\cR$ such that $\X$ is
sound and {\em complete} for $\cR$.  This highlights the importance of (RAA)
in obtaining complete logics.  Our refutation-completeness proof also
implies that the problem of determining whether a given $\cR$-sequent
is valid is in \NLOGSPACE, a fact which is otherwise not obvious.

\subsection{There are no sound and complete 
syllogistic systems for $\cR$}

\begin{theorem}
There exists no finite set $\X$ of syllogistic rules in $\cR$ such that
$\vdash_\X$ is both sound and complete.
\label{theo:noSyllR}
\end{theorem}
\begin{proof}
Let $\X$ be any finite set of syllogistic rules for $\cR$, and suppose
$\vdash_\X$ is sound. We show that it is not complete. Since $\X$ is
finite, fix $n \in \N$ greater than the number of antecedents in any
rule in $\X$.

Let $p_1, \ldots, p_n$ be distinct unary atoms and $r$ a binary atom.
Let $\Gamma$ be the following set of $\cR$-formulas:
\begin{align}
& \forall(p_i, \exists(p_{i+1},r)) & (1 \leq i <  n)
\label{eq:larryProof1} \\
& \forall (p_1, \forall(p_n,r))
\label{eq:larryProof2} \\
& \forall(p,p)  &  (p \in \bP)
\label{eq:larryProof3} \\
& \forall(p_i,\bar{p}_j)  & (1 \leq i < j \leq n)
\label{eq:larryProof4}
\end{align}
and let $\gamma$ be the $\cR$-formula $\forall(p_1,
\exists(p_n,r))$. Observe that $\Gamma \models \gamma$. 
To see this, let $\fA\models\Gamma$.    If
$p_1^\fA = \emptyset$, then trivially $\fA \models \gamma$; on the
other hand, if   $p_1^\fA \neq \emptyset$, a
simple induction using formulas~\eqref{eq:larryProof1} shows that
$p_i^\fA \neq \emptyset$ for all $i$ ($1 \leq i \leq n$), whence $\fA
\models \gamma$ by~\eqref{eq:larryProof2}.  

For $1 \leq i < n$, let $\Delta_i = \Gamma \setminus
\set{\forall(p_i,\exists(p_{i+1},r))}$.
\begin{claim}
If $\phi \in \cR$ and $\Delta_i \models \phi$, then $\phi \in \Gamma$.
\label{claim:R}
\end{claim}
It follows from this claim that $\Gamma \not \vdash_\X \gamma$. For,
since no rule of $\X$ has more than $n-1$ antecedents, any instance of
those antecedents contained in $\Gamma$ must be contained in
$\Delta_i$ for some $i$.  Let $\delta$ be the corresponding instance
of the consequent of that rule.  Since $\vdash_\X$
is sound, $\Delta_i \models \delta$. By
Claim~\ref{claim:R}, $\delta \in \Gamma$. By induction on the number
of steps in derivations, we see that no derivation from $\Gamma$ leads
to a formula not in $\Gamma$. But $\gamma \not \in \Gamma$.

\begin{proof}[Proof of Claim]
Certainly, $\Delta_i$ has a model, for instance the model $\fA_i$ given by:
\begin{equation}
\begin{array}{l}
\xymatrix@C-1pc@R-1pc{  & &  &   \\
p_1\ar `u[ur] `[rrrrrrrrrr]  [rrrrrrrrrr] 
 \ar[r] & p_2 \ar[r] && \cdots & \ar[r]  & p_i 
& p_{i+1} \ar[r] & &\cdots & \ar[r]  & p_{n}\\
} \\
\end{array}
\label{models}
\end{equation}
Here, $A = \{p_1, \ldots, p_n\}$, $p_j^{\fA_i} = \{p_j\}$ for all $j$
($1 \leq j \leq n$), and $r^{\fA_i}$ is indicated by the arrows. All
other atoms (unary or binary) are assumed to have empty
extensions.  Note that there is no arrow from $p_i$ to $p_{i+1}$.

We consider the various possibilities for $\phi$ in turn and check
that either $\phi \in \Gamma$ or there is a model of $\Delta_i$ in
which $\phi$ is false.

\vspace{0.25cm}

\noindent
(i) $\phi$ is of the form $\forall(p,p)$. Then $\phi \in \Gamma$
by~\eqref{eq:larryProof3}.

\vspace{0.25cm}

\noindent
(ii) $\phi$ is not of the form $\forall(p,p)$, and involves at least
  one unary or binary atom other than $p_1, \ldots, p_n, r$.  In this
  case, it is straightforward to   
  modify $\fA_i$ so as to obtain a
  model $\fA'_i$ of $\Delta_i$ such that $\fA'_i \not \models
  \phi$. Henceforth, then, we may assume that $\phi$ involves no atoms
  other than $p_1, \ldots, p_n, r$.

\vspace{0.25cm}

\noindent
(iii) $\phi$ is of the form $\forall(p_j,p_k)$. If $j = k$, then $\phi \in \Gamma$,
by~\eqref{eq:larryProof3}.  If $j \neq k$, then $\fA_i \not \models
\phi$ by inspection.

\vspace{0.25cm}

\noindent
(iv) $\phi$ is of the form $\forall(p_j,\bar{p}_k)$.  If $j = k$, then
$\fA_i \not \models \phi$, since $p_j^{\fA_i} \neq \emptyset$.  If $j
\neq k$, then $\phi \in \Gamma$, by~\eqref{eq:larryProof4}
and the identification $\forall(p_j,\bar{p}_k) =
\forall(p_k,\bar{p}_j)$.

\vspace{0.25cm}

\noindent
(v) $\phi$ is of the form $\forall(p_j,\forall(p_k, r))$. If $j=1$ and
$k= n$, then $\phi \in \Gamma$, by~\eqref{eq:larryProof2}.  So we may
assume that either $j > 1$ or $k< n$, in which case, $k \neq j+1$
implies $\fA_i \not \models \phi$, by inspection. Hence, we may assume
that $\phi= \forall(p_j,\forall(p_{j+1}, r))$, with $j <n$. Let
$\fB_{i,j}$ be the structure obtained from $\fA_i$ by adding a second
point $b$ to  the interpretation of $p_{j+1}$, and to which $p_j$ is
not related by $r$.  In pictures:
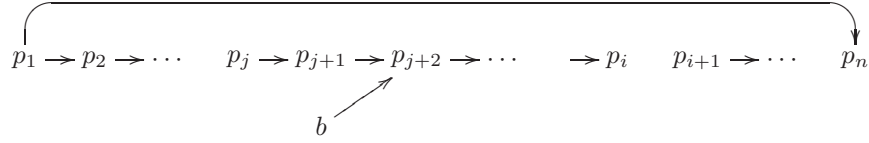
\begin{equation*}
\begin{array}{l}
\xymatrix@C-1.1pc@R-1pc{  & &  &   \\
p_1\ar `u[ur] `[rrrrrrrrrrr]  [rrrrrrrrrrr] 
 \ar[r] & p_2 \ar[r] & \cdots &
 p_j \ar[r] & p_{j+1} \ar[r] & p_{j+2} \ar[r] & \cdots &  \ar[r]  & p_i 
& p_{i+1} \ar[r] &\cdots   & p_{n}\\
 & & & & b \ar[ur]& & &\\
} \\
\end{array}
\end{equation*}
(This picture shows $j+2 < i$. 
 Similar pictures are possible in all other cases.)
By inspection, $\fB_{i,j} \models \Delta_i$, but $\fB_{i,j} \not
\models \phi$.

\vspace{0.25cm}

\noindent
(vi) $\phi$ is of the form $\forall(p_j,\exists(p_k,r))$. If $k =
j+1$, then $\phi \in \Gamma$, by~\eqref{eq:larryProof1}. Moreover, if
$k \neq j+1$, then, unless $j = 1$ and $k=n$, $\fA_i \not \models
\phi$, by inspection.  Hence we may assume $\phi =
\forall(p_1,\exists(p_n,r))$. Let $\fC_i$ be the structure:
\begin{equation*}
\begin{array}{l}
\xymatrix@C-1pc{
p_1
 \ar[r] & p_2 \ar[r] && \cdots & \ar[r]  & p_i, 
} \\
\end{array}
\end{equation*}
with $p_j^{\fC_i} = \emptyset$ for all $j$ ($i < j \leq n$). Then
$\fC_i \models \Delta_i$, but $\fC_i \not \models \phi$.

\vspace{0.25cm}

\noindent
(vii) $\phi$ is of either of the forms
$\forall(p_j,\forall(p_k,\bar{r}))$,
$\forall(p_j,\exists(p_k,\bar{r}))$. Define $\fA''_i$ to be like
$\fA_i$ except that $r^{\fA_i''}$ additionally contains the pair of
points $\langle p_j, p_k \rangle$.  By inspection, $\fA''_i \models
\Delta_i$, but $\fA''_i \not \models \phi$.

\vspace{0.25cm}

\noindent
(viii) $\phi$ is of the form $\exists(p,c)$. Let $\fA_0$ be a
structure over any domain in which every atom has empty extension.
Then $\fA_0 \models \Delta_i$, but $\fA_0 \not \models \phi$.
\end{proof}
This also completes the proof of Theorem~\ref{theo:noSyllR}.
\end{proof}

\subsection{A refutation-complete system for $\cR$}
\label{subsec:R}
Theorem~\ref{theo:noSyllR} notwithstanding, we exhibit below a finite
set $\R$ of rules in $\cR$, such that $\vdash_\R$ is sound and {\em
refutation}-complete.  We remind the reader that $p$ and $q$ range
over unary atoms, $c$ over c-terms, and $t$ over binary literals.
\begin{equation*}
\begin{array}{llll}
\infer[{\mbox{\small (D1)}}]{\exists (p,c)}
                   {\exists (p,q) & 
             \qquad \forall (q,c)}
&
\infer[{\mbox{\small (B)}}]{\forall (p,c)}
                   {\forall (p,q) & 
             \qquad \forall (q,c)}\\
\\
\infer[{\mbox{\small (D2)}}]{\exists (q,c)}
                   {\forall (p,q) & 
             \qquad \exists (p,c)}
&
\infer[{\mbox{\small (T)}}]{\forall (p,p)}{}
\qquad \infer[{\mbox{\small (I)}}]{\exists (p,p)}
                   {\exists (p,c)}    
\\ \\
\infer[{\mbox{\small (D3)}}]{\exists (p,\bar{q})}
                   {\forall (q,\bar{c}) & 
             \qquad \exists (p,c)} 
& \infer[{\mbox{\small (A)}}]{\forall (p,c)}
                   {\forall (p, \bar{p})}
\qquad
\infer[{\mbox{\small (II)}}]{\exists (q, q)}
                            {\exists (p, \exists (q, t))}\\
\ \\
\infer[{\mbox{\small ($\forall\forall$)}}]{\forall (p,\exists(q, t))}
                     {\forall (p, \forall (q', t)) & 
                             \qquad 
                             \exists (q, q')}
\hspace{.5cm}&
\infer[{\mbox{\small ($\exists\exists$)}}]{\exists (p, \exists (q', t))}
                            {\exists (p, \exists (q, t)) & 
                             \qquad 
                             \forall (q, q')}
\\
\ \\
\infer[{\mbox{\small ($\forall\exists$)}}]{\forall (p, \exists (q', t))}
                          {\forall (p, \exists (q, t)) & 
                             \qquad 
                             \forall (q, q')}

\end{array}
\end{equation*}
Rules (D1), (D2), (D3), (B), (A), (T) and (I) are natural
generalizations of their namesakes in $\S$.    In
contrast, ($\forall\forall$), ($\exists\exists$), ($\forall\exists$)
and (II) express genuinely relational logical principles.  In some
settings, these last rules are called \emph{monotonicity principles}.
Because we seek only refutation-completeness for $\vdash_\R$, we do
not need a version of the rule (X).

To illustrate these rules, let $n$ be any integer greater than 1, let
\begin{equation*}
\Gamma^* = \{ \forall(p_i,\exists(p_{i+1}, r)) \mid 1 \leq i < n \} \cup
\{ \forall(p_1,\forall(p_n, r))\}
\end{equation*}
and let $\gamma = \forall(p_1,\exists(p_n, r))$.  Noting that
$\bar{\gamma} = \exists(p_1,\forall(p_n, \bar{r}))$, we have the
derivation (shown here for $n > 3$)
\begin{equation*}
   \infer[{\mbox{\small (D3),}}]
   {\exists (p_1, \bar{p}_1)}
   {\infer[{\mbox{\small ($\forall\forall$)}}]{\forall(p_1,\exists(p_n, r))}
          {\forall(p_1,\forall(p_n, r)) 
           & \hspace{-2cm} \deduce{\exists (p_n, p_n)}
            {\deduce{\ddots}
             {\infer[{\mbox{\small (II)}}]{\exists (p_3, p_3)}
              {\infer[{\mbox{\small (D1)}}]{\exists(p_2, \exists(p_3, r))}
                     {\infer[{\mbox{\small (II)}}]{\exists (p_2, p_2)}
                  {\infer[{\mbox{\small (D1)}}]{\exists(p_1, \exists(p_2, r))}
                     {{\infer[{\mbox{\small (I)}}]{\exists (p_1, p_1)}
                         {\exists(p_1, \forall(p_n, \bar{r}))}}
                      &
                      \forall(p_1, \exists(p_2, r))
                     }
                  }
               &
               \forall(p_2,\exists(p_3, r))
               }
              }
             }
            }
           }
       &
      \hspace{-2cm}  \exists(p_1, \forall(p_n, \bar{r}))
      }
\end{equation*}
showing that $\Gamma^* \cup \{\bar{\gamma}\} \vdash_\R \bot$.  By
contrast, since $\Gamma^* \subseteq \Gamma$, where $\Gamma$ is the set
of formulas used in the proof of Theorem~\ref{theo:noSyllR}, we know
that, for {\em any} finite set  $\X$ of syllogistic 
 rules, $n$ can be 
made
sufficiently large that $\Gamma^* \not \vdash_\X \gamma$.

\paragraph{Starting on the proof}
For the remainder of this section, fix a finite, non-empty set
$\Gamma$ of $\cR$-formulas. As usual, we take the (possibly decorated)
variables $p$, $q$ to range over unary atoms, $r$ over binary atoms,
$t$ over binary literals,
and $c$, $d$ over c-terms. We write $c \Rightarrow d$ if $c
= d$ or there exists a sequence of unary atoms $p_0, \ldots, p_k$ such
that $c = p_0$, $\forall(p_k,d) \in \Gamma$, and $\forall(p_i,p_{i+1})
\in \Gamma$ for all $i$ ($0 \leq i <k$). If $V$ is a set of 
c-terms, write $V \Rightarrow d$ if $c \Rightarrow d$ for some $c
\in V$.
\begin{lemma}
Let $V$ be a set of c-terms.
\begin{enumerate}
\item If $V \Rightarrow c$, then either $c \in V$ or there 
      exists $p \in V$ such that $\Gamma \vdash_\R \forall(p,c)$;
\item if $V \Rightarrow p$, then there exists $p_0 \in V$ such that
      $\Gamma \vdash_\R \forall(p_0,p)$;
\item if $p \Rightarrow c$, then $\Gamma \vdash_\R \forall(p,c)$.
\end{enumerate}
\label{lma:4}
\end{lemma}
\begin{proof}
Almost immediate, noting that $\R$ contains the rules (B) and
(T).
\end{proof}

In the ensuing lemmas, we show that, if $\Gamma$ is consistent (with
respect to $\vdash_\R$), then $\Gamma$ is satisfiable.
As a first step, we create plenty of objects from which to construct a
potential model of $\Gamma$. If $0 \leq i \leq 2$ and $V$ is a set of
c-terms with $1 \leq |V| \leq 2$, let $b_{V,i}$ denote some object or
other, and assume that the various $b_{V,i}$ are pairwise
distinct. Now set
\begin{equation*}
B_0 = \{b_{\{p,c\},0} \mid \exists(p,c) \in \Gamma \}.
\end{equation*}
\begin{lemma}
Let $b_{V,0} \in B_0$, and let $p, c \in V$. Then
$\Gamma \vdash_\R \exists(p,c)$.
\label{lma:5}
\end{lemma}
\begin{proof}
If $p \neq c$, then $V = \{p,c\}$ and $\exists(p,c) \in \Gamma$
by construction. If $p = c$, then $V = \{p,d\}$ for some $d$, and 
we have the derivation
\begin{equation*}
\infer[{\mbox{\small (I)}}.]{\exists(p,p)}{\exists(p,d)}
\end{equation*}
\end{proof}
We now define sets $B_1$, $B_2$, \ldots inductively as follows. Suppose
$B_k$ has been defined.    Let
\begin{equation*}
B_{k+1} = B_k \cup \{b_{\{p\},i} \mid 1 \leq i \leq 2 \mbox{ and, for
    some $b_{V,j} \in B_k$ and some $t$, $V \Rightarrow \exists(p,t)$}\}.
\end{equation*}
Let $B = \bigcup_{0 \leq k} B_k$. Evidently, $B$ is finite.  (Indeed,
$|B|$ is bounded by a linear function of $|\Gamma|$.)  It
is immediate from the construction of $B$ that, if $b_{V,i} \in B_0$,
then $1 \leq |V| \leq 2$, $V$ contains at least one unary atom $p$,
and $i = 0$. On the other hand, if $b_{V,i} \in B_k$ for $k >0$, then
$V = \{ p\}$ for some unary atom $p$, and $i$ is either 1 or 2.  The
intuition here is that the elements of $B_0$ are witnesses for the
existential formulas of $\Gamma$, while the elements of $B_{k+1}$ are
the witnesses for existential c-terms satisfied by elements of $B_k$.

\begin{lemma} 
If $b_{V,i} \in B$, $V \Rightarrow p$ and $V \Rightarrow c$, then
$\Gamma \vdash_\R \exists(p,c)$.
\label{lma:6}
\end{lemma}
\begin{proof}
Let $k$ be the smallest number such that $b_{S,i} \in B_k$. We proceed
by induction on $k$.

For the case $k = 0$, we have $b_{V,i}  = b_{V,0} \in B_0$. By Lemma~\ref{lma:4}
Part 2, there exists $q_1 \in S$ such that $\Gamma \vdash_\R \forall(q_1,p)$.
By Lemma~\ref{lma:4} Part 1, either $c \in V$ or 
there exists $q_2 \in V$ such that $\Gamma \vdash_\R \forall(q_2,c)$. In the
former case, Lemma~\ref{lma:5} yields $\Gamma \vdash_\R \exists(q_1,c)$,
so that we have the derivation
\begin{equation*}
\infer[{\mbox{\small (D2)}}.]{\exists(p,c)}{\infer*{\exists(q_1,c)}{}
                             & \qquad 
                             \infer*{\forall(q_1,p)}{\ }}
\end{equation*}
In the latter case, Lemma~\ref{lma:5} yields $\Gamma \vdash_\R
\exists(q_1,q_2)$,
so that we have the derivation
\begin{equation*}
\infer[{\mbox{\small (D2)}}.]{\exists(p,c)}
     {\infer[{\mbox{\small (D1)}}]{\exists(q_1,c)}
                                   {\infer*{\exists(q_1,q_2)}{}
                                    & \qquad 
                                    \infer*{\forall(q_2,c)}{}} 
                             & \qquad \infer*{\forall(q_1,p)}{}}
\end{equation*}

For the case $k>0$, $b_{V,i} \in B_k$ implies $V = \{ p_k \}$ for some
$p_k$, and $1 \leq i \leq 2$.  By construction of $B_k$, there exist
$b_{W,j} \in B_{k-1}$, $p_{k-1} \in W$ and binary atom $r$, such that $W
\Rightarrow \exists(p_k,r)$.  
By inductive hypothesis, $\Gamma \vdash_\R
\exists(p_{k-1},\exists(p_k,r))$, and by Lemma~\ref{lma:4} Part 3,
$\Gamma \vdash_\R \forall(p_k,p)$, and $\Gamma \vdash_\R \forall(p_k,c)$.
Therefore, we have the derivation
\begin{equation*}
\infer[{\mbox{\small (D1)}}.]{\exists(p,c)}
        {\infer[{\mbox{\small (D1)}}]{\exists(p_k,p)}
                {\infer[{\mbox{\small (II)}}]{\exists(p_k,p_k)}
                        {\infer*{\exists(p_{k-1},\exists(p_k,r))}{}} 
                & \qquad 
                \infer*{\forall(p_k,p)}{}}
                       & \qquad 
                       \infer*{\forall(p_k,c)}{}}
\end{equation*}
\end{proof}
\begin{lemma}
If $b_{V,i} \in B$, $V \Rightarrow c$,  $V \Rightarrow d$, and $c \neq
d$, then there exists a unary atom $p$ such that either:
{\rm (i)}  $\Gamma \vdash_\R \exists(p,c)$ and $\Gamma \vdash_\R \forall(p,d)$; or
{\rm (ii)} $\Gamma \vdash_\R \exists(p,d)$ and $\Gamma \vdash_\R \forall(p,c)$. 
\label{lma:7}
\end{lemma}
\begin{proof}
Suppose first that $c = q$ for some $q$.  By Lemma~\ref{lma:6},
$\Gamma \vdash_\R \exists(q,d)$, and by rule (T), $\Gamma \vdash_\R
\forall(q,q)$.  Putting $p = q$ then satisfies Condition (ii).  On the
other hand, if $d = q$ for some $q$, then Condition (i) is satisfied,
by a similar argument.  Hence we may assume that neither $c$ nor $d$
is a unary atom. Since $c \neq d$, we have either $c \not \in V$ or $d
\not \in V$, by construction of $B$.  If the latter, then, by
Lemma~\ref{lma:4} Part 1, there exists $p$ such that $\Gamma
\vdash_\R \forall(p,d)$. But now we have $V \Rightarrow p$ and $V
\Rightarrow c$, so that, by Lemma~\ref{lma:6}, $\Gamma \vdash_\R
\exists(p,c)$, and Condition (i) is satisfied. If, on the other hand,
$c \not \in V$, Condition (ii) is satisfied, by a similar argument.
\end{proof}

The set $B$ will form the domain of a structure $\fB$, defined as
follows.  If $p$ is a unary atom, set
\begin{equation*}
p^\fB = \{ b_{V,i} \in B \mid V \Rightarrow p \};
\end{equation*}
and if $r$ is a binary atom, set
\begin{multline*}
r^\fB = \{ \langle b_{V,i}, b_{\{p\},1} \rangle \in B^2 \mid V 
                      \Rightarrow \exists(p,r) \} \cup \\
        \{ \langle b_{V,i}, b_{W,j} \rangle \in B^2 \mid  
                      \mbox{ for some $q$, }
                      V \Rightarrow \forall(q,r) \mbox{ and } W \Rightarrow q \}.
\end{multline*}
The intuition is that the elements $b_{\{p\},1}$ are witnesses for the
existential quantifiers in c-terms of the form
$\exists(p,r)$, while the elements $b_{\{p\},2}$ are witnesses for the
existential quantifiers in c-terms of the form
$\exists(p,\bar{r})$.
\begin{lemma}
If $\Gamma$ is unsatisfiable, then there exist an element $b_{V,i}$
of $B$, a unary atom $p$ and a c-term $c$, such that
$V \Rightarrow p$, $V \Rightarrow c$, $b_{V,i} \in p^\fB$ and $b_{V,i}
\not \in c^\fB$.
\label{lma:8}
\end{lemma}
\begin{proof}
Since $\Gamma$ is unsatisfiable, let $\phi \in \Gamma$ be such that
$\fB \not \models \Gamma$.  If $\phi = \exists(p,c)$, let $V =
\{p,c\}$. Trivially, $V \Rightarrow p$ and $V \Rightarrow c$.  By
construction of $B$, $b_{V,0} \in B$, and by construction of $\fB$,
$b_{V,0} \in p^\fB $, whence (since $\fB \not \models \phi$) $b_{V,0}
\not \in c^\fB $. If, on the other hand, $\phi = \forall(p,c)$, there
exists $b_{V,i} \in B$ such that $b_{V,i} \in p^\fB $ and $b_{V,i}
\not \in c^\fB $. By construction of $\fB$, $V \Rightarrow p$; and
since $\forall(p,c) \in \Gamma$, we have $V \Rightarrow c$, as
required.
\end{proof}
We now prove the main Lemma, from which both the complexity and the
refutation-completeness results follow.
\begin{lemma}
If $\Gamma$ is unsatisfiable, then the following condition holds.
\begin{center}
\begin{minipage}{1cm}
$(${\bf C}$)$

\vspace{3cm}

\end{minipage}
\begin{minipage}{10cm}
There exist elements $b_{V,i}, b_{W,j}$ of $B$, unary atoms $q$, $o$,
and binary atom $r$, such that one of the following is true:
\begin{enumerate}
\item $V \Rightarrow q$ and $V \Rightarrow \bar{q}$; 
\item $V \Rightarrow \exists(q,\bar{r})$, $V \Rightarrow \forall(o,r)$, 
  and $q \Rightarrow o$; 
\item $V \Rightarrow \forall(q,\bar{r})$, 
$V \Rightarrow \exists(o,r)$, and $o \Rightarrow q$;
\item $V \Rightarrow \forall(q,\bar{r})$, $V \Rightarrow
  \forall(o,r)$, $W \Rightarrow q$, and $W \Rightarrow o$.
\end{enumerate}
\end{minipage}
\end{center}
\label{lma:9}
\end{lemma}
\begin{proof}
Let $V$, $i$, $p$ and $c$ be as in Lemma~\ref{lma:8}.  We claim first
that $c$ cannot be of the form $q$, $\exists(q,r)$ or
$\forall(q,r)$. For consider each possibility in turn.
\begin{description}
\item If $c = q$, then, by the construction of $\fB$, $V
  \Rightarrow c$ implies  $b_{V,i} \in q^\fB$, contradicting $b_{V,i}
  \not \in c^\fB$.
\item If $c = \exists(q,r)$, then, by the construction of $\fB$,
  $V \Rightarrow c$ implies $b_{\{q\},1} \in B$, $b_{\set{q},1} \in q^\fB$,
  and $\langle b_{V,i}, b_{\{q\},1} \rangle \in r^\fB$, whence
  $b_{V,i} \in \exists(q,r)^\fB$, contradicting $b \not \in c^\fB$.
\item If $c = \forall(q,r)$, then, since $V \Rightarrow c$, the
  construction of $\fB$ ensures that, for any
  $b_{W,j} \in q^\fB$, we have $W \Rightarrow q$ and hence
  $\langle b_{V,i}, b_{W,j} \rangle \in r^\fB$. That is, $b_{V,i} \in
  \forall(q,r)^\fB$, contradicting $b \not \in c^\fB$.
\end{description}
Therefore, $c$ is of one of the forms $\bar{q}$, $\exists(q,\bar{r})$,
or $\forall(q,\bar{r})$. We consider each possibility in turn, and
show that one of the four cases of Condition ({\bf C}) holds.
\begin{description}
\item If $c = \bar{q}$, then $b_{V,i} \not \in c^\fB$ means that $b_{V,i} \in
q^\fB$, so that, by construction of $\fB$, $V \Rightarrow q$. But
  by assumption, $V \Rightarrow c$, and we have Case 1 of Condition ({\bf C}).
\item If $c = \exists(q,\bar{r})$, then, since $V \Rightarrow c$, we
  have, by the construction of $\fB$, $b_{\{q\},2} \in B$, and in fact
  $b_{\{q\},2} \in q^\fB$. Since $b_{V,i} \not \in c^\fB$, we have
  $\langle b_{V,i},b_{\{q\},2} \rangle \in r^\fB$. The construction of
  $\fB$ then guarantees that for some unary atom $o$, $V
  \Rightarrow \forall(o,r)$ and $q \Rightarrow o$; thus we have Case 2
  of Condition ({\bf C}).
\item If $c = \forall (q,\bar{r})$, then $b_{V,i} \not \in c^\fB$ implies
  that there exists $b_{W,j} \in B$ such that $b_{W,j} \in q^\fB$ and
  $\langle b_{V,i}, b_{W,j} \rangle \in r^\fB$.  By construction of
  $\fB$, $W \Rightarrow q$, and, for some unary atom $o$, either
\begin{description}
\item (a) $V \Rightarrow \exists (o,r)$, $W = \{o\}$ and $j= 1$, or
\item (b) $V \Rightarrow \forall (o,r)$ and $W \Rightarrow o$.
\end{description}
 But these yield Cases 3 and 4 of Condition ({\bf C}), respectively.
\end{description}
\end{proof}
\begin{lemma}
The following are equivalent:
\begin{enumerate}
\item $\Gamma \vdash_\R \bot$;
\item $\Gamma$ is unsatisfiable;
\item Condition $(${\bf C}$)$ of Lemma~\ref{lma:9} holds.
\end{enumerate}
\label{lma:tfae}
\end{lemma}
\begin{proof}
For the implication $1 \Rightarrow 2$, we observe that $\vdash_\R$ is
obviously sound.  The implication $2 \Rightarrow 3$ is
Lemma~\ref{lma:9}.  For the implication $3 \Rightarrow 1$, suppose
Condition $(${\bf C}$)$ of Lemma~\ref{lma:9} holds.  This condition
has four cases: we consider each in turn, showing that $\Gamma \vdash_\R
\exists(p,\bar{p})$ for some unary atom $p$.
\begin{enumerate}
\item $V \Rightarrow q$ and $V \Rightarrow \bar{q}$. By
  Lemma~\ref{lma:6} we immediately have $\Gamma \vdash_\R
  \exists(q,\bar{q})$.
\item $V \Rightarrow
  \exists(q,\bar{r})$, $V \Rightarrow \forall(o,r)$, $q \Rightarrow o$. 
By Lemma~\ref{lma:4} Part 2 (or Part~3), $\Gamma \vdash_\R
\forall(q,o)$, and by Lemma~\ref{lma:7}, there exists $p$ such
that either: (i) $\Gamma \vdash_\R \exists(p,\exists(q,\bar{r}))$ and $\Gamma
\vdash_\R \forall(p,\forall(o,r))$; or (ii) $\Gamma \vdash_\R
\exists(p,\forall(o,r))$ and $\Gamma \vdash_\R \forall(p,\exists(q,\bar{r}))$.

In Case (i), we then have
\begin{equation*}
\infer[{\mbox{\small (D3)}},]{\exists(p,\bar{p})}
      {
           \infer[{\mbox{\small ($\exists\exists$)}}]{\exists(p,\exists(o,\bar{r}))}
              {\infer*{\exists(p,\exists(q,\bar{r}))}{} 
               & \qquad 
               \infer*{\forall (q,o)}{}}
       &         \infer*{\forall(p,\forall(o,r))}{} 
               }
\end{equation*}
while in Case (ii), we have
\begin{equation*}
\infer[{\mbox{\small (D3).}}]{\exists(p,\bar{p})}
       {\infer*{\exists(p,\forall(o,r))}{}
       &
        \infer[{\mbox{\small ($\forall\exists$)}}]{\forall(p,\exists(o,\bar{r}))}
              {\infer*{\forall(p,\exists(q,\bar{r}))}{} 
               & \qquad 
               \infer*{\forall (q,o)}{}}
               }
\end{equation*}
\item $V \Rightarrow \forall(q,\bar{r})$, $V \Rightarrow
  \exists(o,r)$, $o \Rightarrow q$.  By Lemma~\ref{lma:4} Part~2 (or
  Part~3), $\Gamma \vdash_\R \forall(o,q)$, and by Lemma~\ref{lma:7},
  there exists $p$ such that either: (i) $\Gamma \vdash_\R
  \exists(p,\exists(o,r))$ and $\Gamma \vdash_\R \forall(p,
  \forall(q,\bar{r}))$; or (ii) $\Gamma \vdash_\R
  \exists(p,\forall(q,\bar{r}))$ and $\Gamma \vdash_\R \forall(p,
  \exists(o,r))$.  But then we can employ exactly the same derivation
  patterns as for Cases 2(i) and 2(ii), respectively.
\item $V \Rightarrow \forall(q,\bar{r})$, $V \Rightarrow
  \forall(o,r)$, $W \Rightarrow q$, $W \Rightarrow o$.  By
  Lemma~\ref{lma:7}, there exists $p$ such that $\Gamma \vdash_\R
  \exists(p,\forall(p_1,u))$ and $\Gamma \vdash_\R
  \forall(p,\forall(p_2,\bar{u}))$, where $u$ is either $r$ or
  $\bar{r}$, and $p_1$ and $p_2$ are $q$ and $o$ in some order. By
  Lemma~\ref{lma:6}, $\Gamma \vdash_\R \exists(o,q)$, i.e.  $\Gamma
  \vdash_\R \exists(p_1,p_2)$.  Thus we have the derivation
\begin{equation}
\infer[{\mbox{\small (D3)}}.]{\exists(p,\bar{p})}
      {\infer[{\mbox{\small ($\forall\forall$)}}]{\forall(p,\exists(p_1,\bar{u}))}
              {\infer*{\forall(p,\forall(p_2,\bar{u}))}{} 
               & \qquad 
               \infer*{\exists (p_1,p_2)}{}}
       & \qquad 
       \infer*{\exists(p,\forall(p_1,u))}{}}
\label{eq:botDerivation}
\end{equation}
\end{enumerate}
\end{proof}
\begin{theorem}
The derivation relation $\vdash_\R$ is 
sound and refutation-complete for 
$\cR$.
\label{theo:Rcompleteness}
\end{theorem}
\begin{proof}
Soundness is obvious. Refutation-completeness is the implication from
2 to 1 in Lemma~\ref{lma:tfae}.
\end{proof}
\begin{theorem}
The problem of determining the validity of a sequent in any of the
fragments $\cS$, $\cS^\dagger$ and $\cR$ is \NLOGSPACE-complete.
\label{theo:Rcomplexity}
\end{theorem}
\begin{proof}
The lower bound is obtained by reduction of the reachability problem
for directed graphs to the validity problem for $\cS$.  If $G=(V,E)$
is a directed graph, take the nodes $V$ to be unary atoms. Let
$\Theta_G = \{ \forall(u,v) \mid (u,v) \in E\}$. It is easy to check
that $\Theta_G \models \forall(u,v)$ if and only if $v$ reachable from
$u$.

Recall that, by the Immerman-Szelepcs\'{e}nyi Theorem, \NLOGSPACE{} is
closed under complementation.  The upper bound for $\cS^\dagger$ is
then immediate, since the problem of determining the satisfiability of
a given set of $\cS^\dagger$-formulas is (almost trivially) reducible
to the problem 2SAT, which is well known to be
\NLOGSPACE-complete. (See, for example,
Papadimitriou~\cite{logic:papadimitriou94}, pp.~185 and~398.)

It remains only to establish the upper bound for $\cR$.   For this, it
suffices to show that the problem of determining the {\em
  un}satisfiability of a given set $\Gamma$ of $\cR$-formulas is
in \NLOGSPACE.  Let $\Gamma$ be a set of $\cR$-formulas. Let $B$ and
$\fB$ be as defined for Lemmas~\ref{lma:5}--\ref{lma:tfae}.
Lemma~\ref{lma:tfae} guarantees that $\Gamma$ is unsatisfiable if and
only if there exist $b_{S,i}, b_{T,j} \in B$, unary atoms $q$, $o$,
and a binary atom $r$ satisfying one of the four cases in Condition
({\bf C}) of Lemma~\ref{lma:9}.  Nondeterministically guess these $V$,
$W$, $i$, $j$, $q$, $o$ and $r$. This requires only logarithmic space,
because only the indices of the relevant atoms need to be encoded, and
the size of $B$ is linear in the number of formulas in $\Gamma$. To
check that $b_{V,i}, b_{W,j} \in B$ is essentially a
graph-reachability problem, as are all the requirements in the four
cases of Condition ({\bf C}). Since graph reachability is
in \NLOGSPACE, this proves the theorem.
\end{proof}

\section{$\cR^*$: an indirect system}
\label{sec:noProofR*}
Next, we consider the fragment $\cR^*$. We refer the reader to
Figure~\ref{fig:syntax} for a review of the syntax of this fragment.  The
complexity-theoretic analysis of $\cR^*$ has been done for us.  It is
contained in the more expressive logic (also inspired by natural
language syntax) investigated in McAllester and
Givan~\cite{logic:mcA+G92}, whose satisfiability problem is shown to
be \NPTIME-complete. Inspection of McAllester and Givan's
\NPTIME-hardness proof ({\em op.~cit.}~pp.~12--14) reveals that it
employs only formulas in our fragment $\cR^*$; hence that the lower
bound applies to $\cR^*$ too. In other words:
\begin{proposition}[McAllester and Givan]
The problem of determining the validity of a sequent in the fragment
$\cR^*$  is co-\NPTIME-complete.
\label{prop:R*hard}
\end{proposition}

Proposition~\ref{prop:R*hard} is significant not least because of what
we know about the complexity of searching for derivations in direct
syllogistic systems.

\begin{corollary}
If \PTIME$\neq$\NPTIME, then there exists no finite set $\X$ of
syllogistic rules in $\cR^*$ such that $\vdash_\X$ is both sound and
refutation-complete.
\label{cor:R*prooftheory}
\end{corollary}
\begin{proof}
Proposition~\ref{prop:R*hard} and Lemma~\ref{lma:PTIME}.
\end{proof}

Let $\R^*$ be the following set of syllogistic rules in $\cR^*$.  Our
result is that the {\em indirect} derivation relation $\Vdash_{\R^*}$ is
complete.  In the following, the variables $b^+,c^+$ range over
positive c-terms, and $d$ over c-terms. (As usual, $p, q$ range over
unary atoms and $r$ over binary atoms.)
\begin{equation*}
\begin{array}{c}
\begin{array}{ccc}
 \infer[{\mbox{\small (T)}}]{\forall (c^+,c^+)}{}
 \qquad
 \infer[{\mbox{\small (I)}}]{\exists (c^+,c^+)}
                   {\exists (c^+,d)}   
 &  &
\infer[{\mbox{\small (B)}}]{\forall (b^+,d)}
                   {\forall (b^+,c^+) & 
             \qquad \forall (c^+,d)}
                      
                   \\
\ \\
\infer[{\mbox{\small (D1)}}]{\exists (b^+,d)}
                   {\exists (b^+,c^+) & 
             \qquad \forall (c^+,d)}    
& & 
\infer[{\mbox{\small (D2)}}]{\exists (c^+,d)}
                   {\forall (b^+,c^+) & 
             \qquad \exists (b^+,d)}
              \\ \\
 \end{array}
 \\
 \begin{array}{ccc}
  \infer[{\mbox{\small (J)}}]{\forall(\forall(q,r),\forall(p,r))}{\forall(p,q)}
&\infer[{\mbox{\small (K)}}]{\forall(\exists(p,r),\exists(q,r))}{\forall(p,q)} 
& \infer[{\mbox{\small (L)}}]{\forall(\forall(p,r),\exists(q,r))}{\exists(p,q)}
\\  \\
 \infer[\mbox{\small(II)}]{\exists(p,p)}{\exists(q, \exists(p,r))}
&   \infer[{\mbox{\small (Z)}}]{\forall (c^+,\forall(p,r))}
                    {\forall (p, \bar{p})}
&
                     \infer[{\mbox{\small (W)}}]{\exists((\forall(p,r),\forall(p,r))}
                    {\forall (p, \bar{p})}
\end{array}
\end{array}
\end{equation*}
Rules (D1), (D2), (B), (T) and (I) are natural generalizations of
their namesakes in $\R$.  Rules (J), (K), and (L) embody
logical principles that are intuitively clear, yet not familiar when
taken as single steps.  If {\sf all porcupines are brown animals},
then {\sf everything which attacks all brown animals attacks all
  porcupines} (J), and {\sf everything which photographs some
  porcupine photographs some brown animal} (K).  And if {\sf some
  porcupines are brown animals}, then {\sf everything which caresses
  all porcupines caresses some brown animals} (L).  We have already
seen (II).  Rule (Z) tells that if there are no porcupines (say), then
all farmers love all porcupines.  Rule (W) tells us that under this
same assumption, there is something which loves all porcupines (simply
because we assume the universe is non-empty, and everything in it
vacuously loves all porcupines).  If we did not assume that the
universe of a model is non-empty, then we would drop (W), and the
completeness of the resulting system would be proved the same way.
 
All the rules of $\R$ are derivable in $\Vdash_{\R^*}$.  For example,
here is a proof of (A):
$$\infer[(\mbox{\small RAA})^1]{\forall(p,d)}{
\infer[(\mbox{\small D1})]{\exists(p,\pbar)}{
\infer[(\mbox{\small I})]{\exists(p,p)}{[\exists(p,\dbar)]^1} & \hspace{.5 cm}   \forall(p,\pbar)}
}
$$
The rest of this section is devoted to a proof that     
$\Vdash_{\R^*}$ is sound and complete for $\cR^*$.

\paragraph{Starting the proof}
To prove completeness, we need only show that every consistent set
$\Gamma$ in the fragment $\cR^*$ is satisfiable.  Also, by
Lemma~\ref{lma:lindenbaum}, we may assume that $\Gamma$ is
$\cR^*$-complete.  For the remainder of this section, fix some
$\cR^*$-complete set of formulas $\Gamma$ which is consistent with
respect to $\Vdash_{\R^*}$.  We simplify our notation to write
$\vdash$ for $\Vdash_{\R^*}$.

We shall construct a structure $\fA$ and prove that it
satisfies $\Gamma$.  First, let $\bC^+$ be the set of positive
c-terms.  Then we define $\fA$ by:
\begin{equation*}
\begin{array}{lcl} A & = & \set{\pair{c_1,c_2, Q} \in
    \bC^+\times\bC^+\times \set{\forall, \exists} : \Gamma
    \proves\exists(c_1,c_2)} \\ \semantics{p} & = &
  \set{\pair{c_1,c_2,Q}\in A: \Gamma \proves \forall(c_1,p) \mbox{ or
    } \Gamma \proves \forall(c_2,p)} \\ 
    \pair{c_1, c_2, Q_1}
  \semantics{r} (d_1,d_2, Q_2) & \mbox{iff} & \mbox{either (a) for some $i$,
    $j$, and $q\in \bP$,} \\ & & \Gamma \proves\forall(c_i,\forall(q,r))
  \mbox{ and $\Gamma\proves\forall(d_j,q)$;} \\ 
& 
 &\mbox{or else (b) $Q_2 = \exists$,   and for some $i$ and $q\in \bP$,}\\
& & d_1 = d_2 = q, \mbox{ and } 
 \Gamma \proves \forall(c_i, \exists(q,r)).\\
  \end{array}
\end{equation*}
\rem{
\begin{equation*}
\begin{array}{lcl} A & = & \set{\pair{c_1,c_2, Q} \in
    \bC^+\times\bC^+\times \set{\forall, \exists} : \Gamma
    \proves\exists(c_1,c_2)} \\ \semantics{p} & = &
  \set{\pair{c_1,c_2,Q}\in A: \Gamma \proves \forall(c_1,p) \mbox{ or
    } \Gamma \proves \forall(c_2,p)} \\ \pair{c_1, c_2, Q}
  \semantics{r} (d_1,d_2, \forall) & \mbox{iff} & \mbox{for some $i$,
    $j$, and $q$,} \\ & & \Gamma \proves\forall(c_i,\forall(q,r))
  \mbox{ and $\Gamma\proves\forall(d_j,q)$} \\ 
\pair{c_1, c_2, Q} \semantics{r} \pair{d_1,d_2, \exists} & 
    \mbox{iff} &\mbox{either }
  \pair{c_1, c_2, Q} \semantics{r} \pair{d_1,d_2, \forall} \mbox{ from
    above}; \\ 
& & \mbox{or else $d_1 = d_2$, and for some $i$,} \\ & &
  \Gamma \proves \forall(c_i, \exists(d_1,r)).\\
  \end{array}
\end{equation*}
}
Note that the set $A$ is non-empty.  For let $p\in\bP$.  If
$\Gamma\proves\exists(p,p)$, then $\pair{p,p,\forall} \in A$.
Otherwise, $\Gamma\proves\forall(p,\pbar)$, and so for all binary
atoms $r$, $\Gamma\proves\exists(\forall(p,r),\forall(p,r))$ by (W).
Thus $\pair{c,c,\forall}\in A$, where $c$ is $\forall(p,r)$.
 
\begin{lemma}
For all $c\in\bC^+$,
$$\semantics{c}  =    \set{\pair{d_1,d_2,Q}\in A: 
\mbox{either } \Gamma \proves \forall(d_1,c)
\mbox{, or } \Gamma \proves \forall(d_2,c)
}.$$
\label{lemma-rstar}
\end{lemma}   

\begin{proof}   
The result for $c$ a unary atom is immediate.  We often shall use the
resulting fact that if $\Gamma\proves \exists(p,p)$, then
$\semantics{p} \neq\emptyset$; it contains both $\pair{p,p,\forall}$
and $\pair{p,p,\exists}$.  The main work concerns c-terms of the form
$\forall(p,r)$ and $\exists(p,r)$.  We remark that all c-terms
referred to in this proof are {\em positive}.

We begin with $c = \forall(p,r)$.  Let $\pair{d_1,d_2,Q} \in
\semantics{\forall(p,r)}$.  Now, either $\Gamma\proves \exists (p,p)$
or $\Gamma \not \proves \exists (p,p)$. If the former, then
$\pair{p,p,\forall}\in \semantics{p}$.  By the semantics of our
fragment, $\pair{d_1,d_2,Q} \semantics{r}\pair{p,p,\forall}$.  By the
structure of $\fA$, there are $i$ and $q$ giving the derivation from
$\Gamma$ as in the tree on the left below:
\begin{equation*}
  \infer[(\mbox{\small B})]{\forall(d_i,\forall(p,r)) }
 { \infer*{\forall(d_i,\forall(q,r))}{}
  &
 \hspace{.5 cm}    \infer[(\mbox{\small J})]{\forall(\forall(q,r),\forall(p,r))}{\infer*{\forall(p,q)}{}}
}
\qquad\quad
   \infer[{\mbox{\small (Z)}.}]{\forall (d_j,\forall(p,r))}
                    {\infer*{\forall (p, \bar{p})}{}}
\end{equation*}
This shows that $\Gamma \proves\forall(d_i,\forall(p,r)) $.  On the
other hand, if $\Gamma\not\proves \exists(p,p)$, we use the assumption
that $\Gamma$ is complete to assert that
$\Gamma\proves\forall(p,\pbar)$.  And then we have the derivation from
$\Gamma$ on the right above, for both $j$.

Conversely,  fix $i$ and suppose that 
 $\Gamma\proves \forall(d_i,\forall(p,r))$.
We claim that 
 $\pair{d_1,d_2,Q}$ belongs to $\semantics{\forall(p,r)}$.
 For this, take any $\pair{b_1,b_2,Q'}\in \semantics{p}$ so that
  $\Gamma\proves \forall(b_j,p)$  for some $j$.
  (We are thus using $b_1$ and $b_2$ to range over positive c-terms,
just as the $c$'s and $d$'s do.)
Then $p$, 
$i$ and $j$ show that  $\pair{d_1,d_2,Q}  \semantics{r}\pair{b_1,b_2, Q'}$.
 This for
   all  elements of $\semantics{p}$
 shows that  
  $\pair{d_1,d_2,Q}  \in \semantics{\forall(p,r)}$.
  
 We next prove the statement of our lemma for $c = \exists(p,r)$.
\newcommand{\one}
  {\infer[(\mbox{\small L})]{\forall(\forall(q,r),\exists(b_k,r))}{\zero}}
\newcommand{\zero}{\infer[(\mbox{\small D1})]{\exists(b_k,q)}
  {\infer{\exists(b_k,b_j)}{\infer*{\exists(b_1,b_2)}{}} & \hspace{.5cm} 
   \infer*{\forall(b_j,q)}{}}}

\newcommand{\two}{\infer[(\mbox{\small B})]{\forall(d_i,\exists(b_k,r))}
                 {\infer*{\forall(d_i,\forall(q,r))}{}
                  &  \hspace{.5 cm}  \one}}
\newcommand{\three}{\infer[(\mbox{\small K})]{\forall(\exists(b_k,r),\exists(p,r))}
        {\infer*{\forall(b_k,p)}{}}}

  \newcommand{\zerofive}
{\infer[(\mbox{\small D1})]{\exists(q,p)}
{\zero & \hspace{.5 cm} 
\infer*{\forall(b_k,p)}{}}}
\newcommand{\zeroseven}
{\infer[(\mbox{\small L})]{\forall(\forall(q,r),\exists(p,r))}{\zerofive}}

\newcommand{\zeroeight}
{\infer[(\mbox{\small B})]{\forall(d_i,\exists(p,r))}
  {\infer*{\forall(d_i,\forall(q,r))}{}   & \hspace{.5 cm} \zeroseven}}

Let $\pair{d_1,d_2,Q} \in \semantics{\exists(p,r)}$.  Thus we have
$\pair{d_1,d_2,Q} \semantics{r}\pair{b_1,b_2,Q'}$ for some
$\pair{b_1,b_2,Q'}\in\semantics{p}$.  We first consider case (a) in
the definition of our structure $\fA$: there are $i$, $j$, and $q$ so
that $\Gamma \proves\forall(d_i,\forall(q,r))$ and
$\Gamma\proves\forall(b_j,q)$.  We have
$\Gamma\proves\exists(b_1,b_2)$, since $\pair{b_1,b_2,Q'}\in A$.
Further, let $k$ be such that $\Gamma\proves\forall(b_k,p)$.  We show
the desired conclusion using a derivation from $\Gamma$:
\begin{equation*}
 \zeroeight
\end{equation*}

This concludes the work in case (a).
In case (b),   
 $Q' = \exists$,  there is some $q\in\bP$ such that
$b_1 = b_2 = q$,
 and   for some $i$,  $\Gamma \proves \forall(d_i, \exists(q,r))$.
 Again we have
$\Gamma \proves\forall(q, p)$.  So we have a derivation from
$\Gamma$ as follows:
\begin{equation*}
\infer[(\mbox{\small B})]{\forall(d_i,\exists(p,r))}
{\infer*{\forall(d_i,\exists(q,r))}{}
 \qquad
 \infer[(\mbox{\small K})]{\forall(\exists(q,r),\exists(p,r))}{
\infer*{\forall(q,p)}{}
}}
\end{equation*}

At this point, we know that if $\pair{d_1,d_2,Q} \in \semantics{\exists(p,r)}$,
then $\Gamma\proves \forall(d_i,\exists(p,r))$ for some $i$.
We now verify the converse.    
Let $\pair{d_1,d_2,Q}\in A$, and fix $i$ such that
$\Gamma\proves \forall(d_i,\exists(p,r))$.  Then $\Gamma\proves
\exists(d_1,d_2)$.  We thus have a derivation from $\Gamma$:
 \begin{equation*}
\infer[(\mbox{\small II})]{\exists(p,p)}{
\infer[(\mbox{\small D1})]{\exists(d_i,\exists(p,r))}
{
\infer[(\mbox{\small I})]{\exists(d_i,d_i)}{\infer*{\exists(d_1,d_2)}{}}
 &
\hspace{.5 cm}   \infer*{\forall(d_i,\exists(p,r))}{}
}
}
\end{equation*}
This goes to show that $\pair{p,p,\exists}\in A$.  
By the construction of $\fA$,
$\pair{d_1,d_2,Q}\semantics{r}\pair{p,p,\exists}$, and
$\pair{p,p,\exists} \in p^\fA$.  So $\pair{d_1,d_2,Q}\in
\semantics{\exists(p,r)}$.  This completes the proof.
\end{proof}                     
\begin{lemma}  
$\fA\models \Gamma$.
\label{lemma-second-McA-G}
\end{lemma}
\begin{proof} 
The proof is by cases on the various formula types in $\cR^*$.  Using
the fact that formulas $\exists(e,f)$ and $\exists(f,e)$ are
identified, and similarly for $\forall(\bar{e},\bar{f})$ and
$\forall(f,e)$, we may take all $\cR^*$-formulas to have one of the
forms:
\begin{equation*}
\forall(c^+,d^+), \hspace{0.5cm}
\forall(c^+,\overline{d^+}), \hspace{0.5cm}
\exists(c^+,d^+), \hspace{0.5cm}
\exists(c^+,\overline{d^+}),
\end{equation*}
where $c^+$ and $d^+$ range over \emph{positive} c-terms. In the
remainder of the proof, we omit the ${}^+$-superscripts for clarity:
i.e.~$c$ and $d$ range over positive c-terms.
 
Let $\phi\in\Gamma$ be $\forall(c,d)$.  Using (B) and
Lemma~\ref{lemma-rstar}, we see that
$\semantics{c}\subseteq\semantics{d}$.

Let $\phi\in \Gamma$ be $\forall(c,\bar{d})$.  Suppose towards a
contradiction that $\fA\not\models\phi$.  Let $\pair{b_1,b_2,Q}\in
\semantics{c}\cap \semantics{d}$.  Let $i$ and $j$ be such that
$\Gamma\proves \forall(b_i,c)$ and $\Gamma\proves \forall(b_j,d)$.
Then using (B), $\Gamma\proves \forall(b_i,\bar{d})$.  And since
$\Gamma\proves\exists(b_i,b_j)$, we use (D1) to see that
$\Gamma\proves\exists(d,\bar{d})$.   So $\Gamma$ is
inconsistent, a contradiction.

If $\phi\in\Gamma$ 
is   $\exists(c,d)$,
then $(c,d,\exists)\in A$.
Indeed, $(c,d,\exists)\in \semantics{c}\cap \semantics{d}$, by Rule
(T)
and Lemma~\ref{lemma-rstar}.

Finally,  consider the case when $\phi\in\Gamma$ 
is of the form $\exists(c,\dbar)$.
Then, using (I), $\Gamma\proves\exists(c,c)$,
so $\pair{c,c,\forall}\in A$. 
Suppose towards a contradiction that $\fA\models\forall(c,d)$.
Then $\pair{c,c,\forall}\in \semantics{d}$.
But then we have $\Gamma\proves\forall(c,d)$, by Lemma~\ref{lemma-rstar} again.
One application of (D2) now shows that  $\Gamma\proves\exists(d,\dbar)$.
Thus we have a contradiction to the consistency of $\Gamma$.
\end{proof}

Hence, we have shown that any consistent set $\Gamma$ of
$\cR^*$-formulas has a model. Therefore:
\begin{theorem}
The derivation relation $\Vdash_{\R^*}$ is sound and complete for $\cR^*$.
\label{theo:rstarcompleteness}
\end{theorem}
 
\section{$\cR^\dagger$ and $\cR^{*\dagger}$: no indirect systems}
\label{sec:noProofR+}

Our last fragments are $\cR^\dagger$ and $\cR^{*\dagger}$.  We open
with some simple complexity results, showing that these fragments have
no direct syllogistic proof-system which is both sound and
refutation-complete.  After this, we strengthen the result to show
that there are no {\em in}direct syllogistic proof-systems, either.
The proof of this negative result is similar to what we saw earlier in
Theorem~\ref{theo:noSyllR} for $\cR$, but the argument here is more
intricate.

\begin{lemma}
The problem of determining the validity of a sequent in $\cR^\dagger$ is
\EXPTIME-hard.
\label{lma:R+hard}
\end{lemma}
\begin{proof}
The logic $K^U$ is the basic modal logic $K$ together with an
additional  modality $U$ (for ``universal''), 
whose semantics are given by the
standard relational (Kripke) semantics, plus
\begin{center}
$\models_w U\phi$ if and only if $\models_{w'} \phi$ for all worlds $w'$.
\end{center}
The satisfiability problem for $K^U$ is \EXPTIME-hard. (The proof is an
easy adaptation of the corresponding result for propositional dynamic
logic; see, e.g.~Harel {\em et al.}~\cite{logic:hkt00}:~216~ff.)  It
suffices, therefore, to reduce this problem to satisfiability in
$\cR^\dagger$. Let $\phi$ be a formula of $K^U$.

We first transform $\phi$ into an equisatisfiable set of formulas
$T_\phi \cup S_\phi$ of first-order logic; then we translate the
formulas of $T_\phi \cup S_\phi$ into an equisatisfiable set of
$\cR^\dagger$-formulas.  To simplify the notation, we shall take unary
atoms (in $\cR^\dagger$) to be unary predicates (in first-order
logic); similarly, we take binary atoms to do double duty as binary
predicates.  Let $r$ and $e$ be binary atoms.  For any $K^U$-formula
$\psi$, let $p_\psi$ be a unary atom, and define the set of
first-order formulas $T_\psi$ inductively as follows:
\begin{eqnarray}
\nonumber
T_p & = & \emptyset \text{ (where $p$ is a proposition letter)}\\
\nonumber
T_{\psi \wedge \pi} & = &
\begin{split}
T_\psi \cup T_\pi \cup \{&
  \forall x(p_\psi(x) \wedge p_\pi(x) \rightarrow p_{\psi \wedge \pi}(x)),\\
& \forall x( p_{\psi \wedge \pi}(x) \rightarrow p_\psi(x)),
 \forall x( p_{\psi \wedge \pi}(x) \rightarrow p_\pi(x)) \}
\end{split}
\\
\nonumber
T_{\neg \psi} & = & T_\psi \cup \{
    \forall x( p_{\neg \psi}(x) \rightarrow \neg p_\psi(x)),
    \forall x( \neg p_{\neg \psi}(x) \rightarrow p_\psi(x)) \}\\
\nonumber
T_{\Box \psi} & = & 
\begin{split} T_\psi \cup \{&
    \forall x( p_{\Box \psi}(x) \rightarrow \forall y (\neg p_\psi(y)
    \rightarrow \neg r(x,y))), \\
  &  \forall x( \neg p_{\Box \psi}(x) \rightarrow \exists y (\neg p_\psi(y)
    \wedge r(x,y))) \} 
\end{split} \\
\nonumber
T_{U \psi} & = & 
\begin{split}
T_\psi \cup \{&
    \forall x( p_{U \psi}(x) \rightarrow \forall y (\neg p_\psi(y)
    \rightarrow \neg e(x,y))),\\
  &  \forall x( \neg p_{U \psi}(x) \rightarrow \exists y (\neg p_\psi(y)
    \wedge e(x,y))) \}.
\end{split}
\end{eqnarray}
Now let $S_\phi$ be the collection of five first-order formulas
\begin{equation*}
\exists x(p_{\phi}(x) \wedge p_{\phi}(x)), \hspace{1cm}
  \forall x (\pm p_\phi(x) \rightarrow \forall y (\pm p_\phi(y) \rightarrow
  e(x,y))) .
\end{equation*}
(Although the first formula looks like it has a redundant conjunct, we
state it in this way only to make our work below a little easier.)  We
claim that the modal formula $\phi$ is satisfiable if and only if the
set of first-order formulas $T_\phi \cup S_\phi$ is satisfiable.  For
let $\fM$ be any (Kripke) model of $\phi$ over a frame $(W,R)$. Define
the first-order structure $\fA$ with domain $W$, by setting $r^\fA =
R$, $e^\fA = A^2$, and $p_\psi^\fA = \{ w \mid \fM \models_w \psi\}$,
for any subformula $\psi$ of $\phi$.  It is then easy to check that
$\fA \models T_\phi \cup S_\phi$. Conversely, suppose $\fA \models
T_\phi \cup S_\phi$. We build a Kripke structure $\fM$ over the frame
$(A, r^\fA)$ by setting, for any proposition letter $o$ mentioned in
$\phi$, $\fM \models_a o$ if and only if $a \in
p_o^\fA$. A straightforward structural induction
establishes that for any subformula $\psi$ of $\phi$,  
$\fM \models_a \psi$ if and only if $a \in p_\psi^\fA$. The
formula $\exists x (p_\phi(x) \wedge p_\phi(x)) \in S_\phi$ then
ensures that $\phi$ is satisfied in $\fM$.

Now, all of the formulas in $T_\phi \cup S_\phi$ are of one of the
forms
\begin{align}
   &  \forall x (\pm p(x) \rightarrow \pm q(x)) & & 
      \forall x (\pm p(x) \rightarrow \forall y(\pm q(y) \rightarrow \pm r(x,y)))
\label{eq:form1} \\
   &  \exists x (p(x) \wedge p(x)) & &
      \forall x (\pm p(x) \rightarrow \exists y(\pm q(y) \wedge r(x,y)))
\label{eq:form2} \\
   &  \forall x (p(x) \wedge q(x) \rightarrow o(x)).
\label{eq:form3}
\end{align}
Notice that formulas of the forms~\eqref{eq:form1}
and~\eqref{eq:form2} translate (in the obvious sense) directly into
the fragment $\cR^\dagger$; those of form~\eqref{eq:form3}, by
contrast, do not. The next step is to eliminate formulas of this last
type.

Let $o^*$ be a new unary relation symbol.  For $\theta \in T_\phi \cup S_\phi$
of the form~\eqref{eq:form3}, let $r_\theta$ be a new binary
atom, and define $R_\theta$ to be the set of formulas
\begin{align}
\forall x (\neg o(x) \rightarrow \exists z (o^*(z) \wedge
r_\theta(x,z)))
\label{eq:sim1}\\
\forall x (p(x) \rightarrow \forall z (\neg p(z) \rightarrow \neg r_\theta(x,z)))
\label{eq:sim2}\\
\forall x (q(x) \rightarrow \forall z (p(z) \rightarrow \neg r_\theta(x,z))),
\label{eq:sim3}
\end{align}
which are all of the forms in~\eqref{eq:form1} or~\eqref{eq:form2}. It
is easy to check that $R_\theta \models \theta$. For suppose (for
contradiction) that $\fA \models R_\theta$ and $a$ satisfies $p$ and
$q$ but not $o$ in $\fA$. 
By~\eqref{eq:sim1}, there exists $b$ such
that $\fA \models r_\theta[a,b]$. If $\fA \not \models p[b]$, then~\eqref{eq:sim2}
is false in $\fA$; on the other hand, if $\fA \models p[b]$,
then~\eqref{eq:sim3} is false in $\fA$.  Thus, $R_\theta \models
\theta$ as claimed. Conversely, if $\fA \models \theta$, expand $\fA$
to a structure $\fA'$ by interpreting $o^*$ and $r_\theta$ as follows:
\begin{eqnarray*}
(o^*)^\fA & = & A \\
r_\theta^\fA & = & \{ \langle a,a \rangle \mid \fA \not \models o[a] \}.
\end{eqnarray*}
We check that $\fA' \models R_\theta$. Formula~\eqref{eq:sim1} is
true, because $\fA' \not \models o[a]$ implies $\fA'\models
r_\theta[a,a]$. Formula~\eqref{eq:sim2} is true, because $\fA'\models
r_\theta[a,b]$ implies $a = b$.  To see that Formula~\eqref{eq:sim3}
is true, suppose $\fA'\models q[a]$ and $\fA'\models p[b]$.  If $a =
b$, then $\fA \models o[a]$ (since $\fA' \models \theta$); that is,
either $a \neq b$ or $\fA \models o[a]$. By construction, then, $\fA'
\not \models r_\theta[a,b]$.

Now let $T^*_\phi$ be the result of replacing all formulas $\theta$ in
$T_\phi$ of form~\eqref{eq:form3} with the corresponding trio
$R_\theta$. (The binary atoms $r_\theta$ for the various $\theta$ are
assumed to be distinct; however, the same unary atom $o^*$ can be used
for all $\theta$.)  By the previous paragraph, $T^*_\phi \cup S_\phi$
is satisfiable if and only if $T_\phi \cup S_\phi$ is satisfiable, and
hence if and only if $\phi$ is satisfiable.  But $T^*_\phi \cup
S_\phi$ is a set of formulas of the forms~\eqref{eq:form1}
and~\eqref{eq:form2}, and can evidently be translated into a set of
$\cR^\dagger$-formulas satisfied in exactly the same
structures. Moreover, this set can be computed in time bounded by a
polynomial function of $\lVert \phi \rVert$. This completes the
reduction.
\end{proof}
We note the following fact. (We omit a detailed proof, since
subsequent developments do not hinge on this result.)
\begin{lemma}
The problem of determining the validity of a sequent in 
$\cR^{*\dagger}$ is in \EXPTIME.
\label{lma:R*+complexity}
\end{lemma}
\begin{proof}
Trivial adaptation of Pratt-Hartmann~\cite{logic:ph04}, Theorem~3, which
considers a fragment obtained by adding relative clauses to the
relational syllogistic.
\end{proof}
\begin{theorem}
The problem of determining the validity of a sequent in either of
the fragments $\cR^\dagger$ and $\cR^{*\dagger}$ is \EXPTIME-complete.
\label{theo:R+complexity}
\end{theorem}
\begin{proof}
Lemmas~\ref{lma:R+hard} and~\ref{lma:R*+complexity}.
\end{proof}
\begin{corollary}
There exists no finite set $\X$ of syllogistic rules in either $\cR^\dagger$
or $\cR^{*\dagger}$ such that $\vdash_\X$ is both sound and
refutation-complete.
\label{cor:noSyllR+R*+}
\end{corollary}
\begin{proof}
It is a standard result that \PTIME$\neq$\EXPTIME.  The result is then
immediate by Lemmas~\ref{lma:PTIME} and~\ref{lma:R+hard}.
\end{proof}

Of course, Corollary~\ref{cor:noSyllR+R*+} leaves open the possibility
that there exist {\em indirect} syllogistic systems that are sound and
complete for $\cR^\dagger$ and $\cR^{*\dagger}$. To show that there do not,
stronger methods are required. That is the task of the remainder of
this section.

\paragraph{Starting on the proof}
Let $n \geq 2$; let $o_1, o_2, o_3, q_1, q_2, p_1, \ldots, p_n$ be
unary atoms, and $r, s$ binary atoms; and let $A^{(n)}$ be the set
$\{a_1, \ldots, a_n, a'_1, \ldots, a'_n, u_0,$ $u_1, u_2, u_3, u_4,
v_1, v_2\}$.  We take the structure $\fA^{(n)}$, with domain
$A^{(n)}$, to be as depicted in Fig.~\ref{fig:rPlusProof}. Here,
membership of an element in $p^{\fA^{(n)}}$ (where $p$ is any unary
atom) is indicated by writing $p$ in a box next to that element. The
extensions of the binary atoms $r$ and $s$ are depicted similarly: we
indicate that a pair of elements is in ${r}^{\fA^{(n)}}$ by writing
$r$ next to an arrow between those elements (and likewise for
$s$). The $r$-labelled arrow from the dotted box to the element $u_2$
is to be interpreted as follows: for any element $a$ inside the dotted
box, $\langle a, u_2 \rangle \in r^{\fA^{(n)}}$.
\begin{figure}
\centerline{\input{rPlusProof.pstex_t}}
\caption{The structure $\fA^{(n)}$. Every element inside the dotted
box is related by $r$ to $u_2$.}
\label{fig:rPlusProof}
\end{figure}
Thus, $\fA^{(n)}$ contains two `$r$-chains' of elements satisfying,
successively, $p_1, \ldots, p_n$: the elements of the first chain
additionally satisfy $q_1$ but not $q_2$; those of the second chain
additionally satisfy $q_2$ but not $q_1$. The terminal elements of the
two $r$-chains are related, by $s$, to the elements $v_1$ and $v_2$.
There are also some additional `out-lier' elements, which ensure the
truth of various existential formulas in $\fA$: for example, $u_1$
ensures (together with $u_0$) 
the truth of $\exists(o_1, \exists(p_1, \bar{r}))$; $u_2$
ensures the truth of $\exists(q_1,q_2)$, and so on. In the sequel, we
define structures containing additional elements, which make {\em
  almost} the same $\cR^{*\dagger}$-formulas true as $\fA^{(n)}$. These
latter structures will then be used to show that there can be no sound
and complete indirect syllogistic system for $\cR^*$ or $\cR^{*\dagger}$.

To avoid notational clutter, we drop the superscripts in $\fA^{(n)}$
(and in related constructions) where the value of $n$ does not
matter. By inspection, we have
\begin{align}
& \fA \models \forall(o_1, \exists(q_1,r))
\label{eq:A1} \\
& \fA \models \forall(o_1, \forall(\bar{p}_1,\bar{r}) )
\label{eq:A2} \\
& \fA \models \forall(q_1, \forall(\bar{q}_1,\bar{r}))
\label{eq:A3} \\
& \fA \models \forall(q_2, \forall(\bar{q}_2,\bar{r}))
\label{eq:A4} \\
& \fA \models \forall(p_i, \exists(p_{i+1},r)) \quad (1 \leq i < n)
\label{eq:A5} \\
& \fA \models \forall(q_1, \forall(o_3,\bar{s}))
\label{eq:A6} \\
& \fA \models \forall(q_2, \forall(\bar{o}_3,\bar{s}))
\label{eq:A7} \\
& \fA \models \forall(p_n, \exists(o_2,s)).
\label{eq:A8} 
\end{align}
\noindent

For any structure $\fC$, denote by ${\rm Th}(\fC)$ the set of
$\cR^\dagger$-formulas true in $\fC$, and denote by ${\rm Th}^*(\fC)$ the
set of $\cR^{*\dagger}$-formulas true in $\fC$. Recall that a set of
formulas $\Phi$ is $\cR^\dagger$-complete (or $\cR^{*\dagger}$-complete) if, for
every $\cR^\dagger$-formula (respectively, $\cR^{*\dagger}$-formula) $\phi$,
either $\phi \in \Phi$ or $\bar{\phi} \in \Phi$.  Trivially, for any
$\fC$, ${\rm Th}(\fC)$ is $\cR^\dagger$-complete, and ${\rm Th}^*(\fC)$ is
$\cR^{*\dagger}$-complete.

Let $\gamma$ be the $\cR^\dagger$-formula given by
\begin{equation*}
\gamma =  \forall(o_1, \exists(\bar{q}_2,r)). 
\end{equation*}
Noting that $o_1^\fA = \{u_0\}$, $\langle u_0,a_1 \rangle \in
r^\fA$, and $a_1 \not \in q_2^\fA$, we have $\fA \models \gamma$ (i.e.
$\gamma \in {\rm Th}(\fA)$). Let $\Gamma^{(n)}$ be obtained from ${\rm
Th}(\fA^{(n)})$ by reversing the truth-value of $\gamma$, and similarly for
$\Gamma^{*(n)}$. That is:
\begin{eqnarray*}
\Gamma^{(n)} & = & \left( {\rm Th}\left(\fA^{(n)}\right) \setminus 
                          \{ \gamma \} \right) \cup \{ \bar{\gamma} \}\\
\Gamma^{*(n)} & = & \left( {\rm Th}^*\left(\fA^{(n)}\right) \setminus 
                          \{ \gamma \} \right) \cup \{ \bar{\gamma} \}.
\end{eqnarray*}
Again, we drop the $(n)$-superscript when the value of $n$ does not
matter.  Thus, $\Gamma$ is $\cR^\dagger$-complete and $\Gamma^*$ is
$\cR^{*\dagger}$-complete.
\begin{lemma}
$\Gamma$ \textup{(}and therefore $\Gamma^*$\textup{)} is unsatisfiable.
\label{lma:GammaUnsat}
\end{lemma}
\begin{proof}
Noting that $\bar{\gamma}$ is the formula $\exists(o_1, \forall(\bar{q}_2,
\bar{r}))$, from~\eqref{eq:A1} and~\eqref{eq:A2}, we see that, in any
model $\fB \models \Gamma$, there exists $b \in p_1^\fB \cap q_1^\fB
\cap q_2^\fB$. From~\eqref{eq:A3}--\eqref{eq:A5}, it then follows
that, in any model $\fB \models \Gamma$, there exists $b \in p_n^\fB
\cap q_1^\fB \cap q_2^\fB$. But, from~\eqref{eq:A6}--\eqref{eq:A8}, no
model $\fB \models \Gamma$ can have any element $b \in p_n^\fB \cap
q_1^\fB \cap q_2^\fB$. This proves the lemma.
\end{proof}

Fixing $n \geq 2$, for any $i$ ($1 \leq i < n$), let $B^{(n)}_i =
A^{(n)} \cup \{b_1, \ldots, b_i, u_5\}$, and consider the structure
$\fB^{(n)}_i$, with domain $B^{(n)}_i$, depicted in
Fig.~\ref{fig:rPlusProofBi}. We employ the same notational conventions
as in Fig.~\ref{fig:rPlusProof}.  In particular, the $r$-labelled
arrows from the dotted boxes are to be interpreted as follows: for any
element $a$ inside either of the dotted boxes, $\langle a, u_2 \rangle
\in r^{\fB^{(n)}_i}$. Again, we drop superscripts where the value of
$n$ does not matter. The structure $\fB_i$ contains a copy of $\fA$,
but has an additional $r$-chain whose elements satisfy both $q_1$ and
$q_2$; notice, however, that this additional $r$-chain stops at the
$i$th element.
\begin{figure}

\centerline{\input{rPlusProofBi.pstex_t}}
\caption{The structure $\fB_i$. Every element inside either of the dotted
boxes is related by $r$ to $u_2$.}
\label{fig:rPlusProofBi}
\end{figure}

We employ the following terminology. An {\em existential} formula is
one of the form $\exists(e,f)$.  If $\phi = \exists(e,f)$ and $\fC$ 
a structure, then a {\em witness for} $\fC \models \phi$ is any $a \in C$
such that $a \in e^\fC$ and $a \in f^\fC$. Now, we saw above that $\fA
\models \gamma$; by contrast $\fB_i \models \bar{\gamma}$
($\bar{\gamma}$ is an existential formula, and $u_5$ is a
witness). Let $\delta_i$ be the $\cR^\dagger$-formula given by
\begin{equation*}
\delta_i =  \forall (p_i,\exists(p_{i+1},r)). 
\end{equation*}
We observed in~\eqref{eq:A5} that $\fA \models \delta_i$. By contrast,
$\fB_i \models \bar{\delta_i}$ ($\bar{\delta_i}$ is an existential
formula, and $b_i$ is a witness). However, it
turns out that $\gamma$ and $\delta_i$ 
are the {\em only} differences between $\fA$ and
$\fB_i$ as far as $\cR^{*\dagger}$ is concerned:
\begin{lemma}
For all $i$ \textup{(}$1 \leq i < n$\textup{)}:
\begin{equation*}
{\rm Th}\left( \fB^{(n)}_i \right) =
\left( {\rm Th}\left( \fA^{(n)} \right) \setminus \{\gamma, \delta_i\}
\right) \cup \{\bar{\gamma}, \bar{\delta_i}\};
\end{equation*}
similarly,
\begin{equation*}
{\rm Th}^*\left( \fB^{(n)}_i \right) =
\left( {\rm Th}^*\left( \fA^{(n)} \right) \setminus \{\gamma, \delta_i\}
\right) \cup \{\bar{\gamma}, \bar{\delta_i}\}.
\end{equation*}
\label{lma:Bi}
\end{lemma}
\begin{proof}
Since $\gamma, \delta_i, \bar{\gamma}, \bar{\delta_i}$ are all in 
$\cR^\dagger$, the first of these statements follows instantly from
the second. We proceed, therefore, to establish the second statement.

We first prove the statement for the case $i = 1$. Notice that here,
the third $r$-`chain' in $\fB^{(n)}_1$ contains just one element,
$b_1$.  Our initial goal, then, is to show that ${\rm Th}^*(
\fB^{(n)}_1) = ({\rm Th}^*( \fA^{(n)}) \setminus \{\gamma, \delta_1\})
\cup \{\bar{\gamma}, \bar{\delta_1}\}$.  The basic intuition is
simple: on the one hand, we may check the cases $n = 2$ and $n= 3$ by
brute force; on the other, we see from Fig.~\ref{fig:rPlusProofBi}
that, once the ends of the first two $r$-chains are sufficiently
distant from the elements $u_5$ and $b_1$,
extending those $r$-chains further will have no effect
on the differences between the formulas made true in $\fA^{(n)}$ and
$\fB^{(n)}_1$.

To make this idea rigorous, we first establish three simple claims.
We employ the following terminology. An {\em existential} e-term is
one of the form $\exists(l,t)$. (As usual, $l$ ranges over unary
literals, and $t$ over binary literals.)  If $e = \exists(l,t)$, $\fC$
is a structure, and $a \in C$, then a {\em witness for} $a \in e^\fC$
is any $b \in C$ such that $b \in l^\fC$ and $\langle a,b \rangle \in
t^\fC$.
\begin{claim}
Let $a$ be any element of $A^{(n)}$
 \textup{(}$n \geq 2$\textup{)}, 
and $e$ any e-term. Then
$a \in e^{\fB^{(n)}_1}$ if and only if $a \in e^{\fA^{(n)}}$.
\label{claim:theories1}
\end{claim}
\begin{proof}
We prove the only-if direction by induction on $n$.  The if-direction
then follows by considering the e-term $\bar{e}$. The case $n=2$
is checked by brute force. (We used a computer.) Assume the
claim is true for $n = m \geq 2$; we show that it is true for $n =
m+1$. Let $a \in A^{(m+1)}$. For contradiction, suppose $a \in
e^{\fB^{(m+1)}_1}$, but $a \not \in e^{\fA^{(m+1)}}$. Then $e$ is
existential, and $a \in e^{\fB^{(m+1)}_1}$ has a witness in
$B^{(m+1)}_1 \setminus A^{(m+1)}$.  Writing $e = \exists(l,t)$, there
exists $b \in \{u_5, b_1\}$ such that $b \in l^{\fB^{(m+1)}_1}$ and
$\langle a, b \rangle \in t^{\fB^{(m+1)}_1}$. In that case, if $a \in
A^{(m)}$, we obviously have $a \in e^{\fB^{(m)}_1}$, whence $a \in
e^{\fA^{(m)}}$ by inductive hypothesis, whence $a \in
e^{\fA^{(m+1)}}$, since $\fA^{(m)}$ is a sub-model of
$\fA^{(m+1)}$---a contradiction.  Thus, $a$ is either $a_{m+1}$ or
$a'_{m+1}$. But, by inspection of Fig.~\ref{fig:rPlusProofBi} (bearing
in mind $i = 1$), it is easy to see that, for $a \in \{a_{m+1},
a'_{m+1}\}$, if $a \in e^{\fB^{(m+1)}_1}$ with a witness in $\{u_5,
b_1\}$, then $a \in e^{\fA^{(m+1)}}$ with a witness in $\{u_0, u_2,
a_1, a'_1 \}$--- again, a contradiction.
\end{proof}

\begin{claim}
Let $\phi$ be any existential formula other than $\bar{\delta}_1$, and
let $A_0 = \{a_1, a'_1, u_2 \}$.  If $b_1$ is a witness for
$\fB^{(n)}_1 \models \phi$, then there exists $a \in A_0$ such that 
$a$ is a witness for $\fA^{(n)} \models \phi$.
\label{claim:theories2}
\end{claim}
\begin{proof}
By induction on $n$.  The cases $n=2$ and $n=3$ are checked by brute
force. (We used a computer.) Note in passing that the condition that
$\phi$ is not $\bar{\delta}_1$ is required for these cases.  Suppose
now the claim holds for $n = m \geq 3$; we show that it holds for $n =
m+1$. If $e$ is any e-term, let $\hat{e}$ be the result of replacing
any occurrence of the unary atom $p_{m+1}$ by the unary atom $p_m$.
Writing $\phi = \exists(e,f)$, let $\hat{\phi} =
\exists(\hat{e},\hat{f})$.  Note that, since $m \geq 3$ and $\phi \neq
\bar{\delta}_1$, we have $\hat{\phi} \neq \bar{\delta}_1$.  The
following facts are obvious (see Figs.~\ref{fig:rPlusProof}
and~\ref{fig:rPlusProofBi}):
\begin{eqnarray*}
a \in e^{\fA^{(m+1)}} & \mbox{iff} & a \in \hat{e}^{\fA^{(m)}}
 \qquad \text{for all $a \in A_0$}\\
b_1 \in e^{\fB^{(m+1)}_1} & \mbox{iff} & b_1 \in \hat{e}^{\fB^{(m)}_1}.
\end{eqnarray*}
It follows that, if $b_1$ is a witness for $\fB^{(m+1)}_1 \models
\phi$, then $b_1$ is a witness for $\fB^{(m)}_1 \models \hat{\phi}$,
whence, by inductive hypothesis, there exists $a \in A_0$ such that
$a$ is a witness for $\fA^{(m)} \models \hat{\phi}$, whence there exists
$a \in A_0$ such that $a$ is a witness for $\fA^{(m+1)} \models
\phi$. This completes the induction.
\end{proof}
\begin{claim}
Let $\phi$ be any existential formula other than $\bar{\gamma}$, and
let $A_0 = \{u_0, u_2, u_3\}$. If $u_5$ is a witness for 
$\fB^{(n)}_1 \models \phi$, then there 
exists $a \in A_0$ such that $a$ is a
witness for $\fA^{(n)} \models \phi$.
\label{claim:theories3}
\end{claim}
\begin{proof}
Identical in structure to the proof of Claim~\ref{claim:theories2}.
\end{proof}

We can now prove Lemma~\ref{lma:Bi} in the case $i=
1$. Claim~\ref{claim:theories1} shows that any existential formula
true in $\fA^{(n)}$ is true in $\fB^{(n)}_1$, and moreover, that if
$\phi$ is an existential formula true in $\fB^{(n)}_1$ but false in
$\fA^{(n)}$, then either either $u_5$ or $b_1$ is a witness. But these
possibilities are ruled out---except in the cases $\phi =
\bar{\gamma}$ and $\phi = \bar{\delta}_1$---by
Claims~\ref{claim:theories2} and~\ref{claim:theories3}, respectively.
Finally, if $\phi$ is universal (and not equal to $\gamma$ or
$\delta_1$), then $\bar{\phi}$ is existential
(and not equal to $\bar{\gamma}$ or $\bar{\delta_1}$).  The foregoing analysis
shows that $\fA^{(n)}$ and $\fB^{(n)}_1$ agree on the truth value of
$\bar{\phi}$, hence on the truth value of $\phi$.

We have thus proved the lemma for $i= 1$ and all values of $n$.  Fix
$n \geq 3$, and consider now the structure $\fB^{(n)}_2$, which
differs from $\fB^{(n)}_1$ only in that the third $r$-chain has been
extended from one to two elements.  By inspection of
Fig.~\ref{fig:rPlusProofBi}, the only effect on the set of sentences
made true is to restore the truth of $\delta_1$ and to falsify
$\delta_2$. That is:
\begin{eqnarray*}
{\rm Th}^*\left( \fB^{(n)}_{2} \right) & = &
\left( {\rm Th}^*\left( \fB^{(n)}_1 \right) \setminus \{\bar{\delta}_1, \delta_{2}\}
\right) \cup \{\delta_1, \bar{\delta}_{2}\}\\
\ & = & 
\left( {\rm Th}^*\left( \fA^{(n)} \right) \setminus \{\gamma, \delta_{2}\}
\right) \cup \{\bar{\gamma}, \bar{\delta}_{2}\}.
\end{eqnarray*}
Proceeding in the same way, we have, for all $i$ ($1 \leq i < n -1$),
\begin{eqnarray*}
{\rm Th}^*\left( \fB^{(n)}_{i+1} \right) & = &
\left( {\rm Th}^*\left( \fB^{(n)}_i \right) \setminus \{\bar{\delta}_i, \delta_{i+1}\}
\right) \cup \{\delta_i, \bar{\delta}_{i+1}\}\\
\ & = & 
\left( {\rm Th}^*\left( \fA^{(n)} \right) \setminus \{\gamma, \delta_{i+1}\}
\right) \cup \{\bar{\gamma}, \bar{\delta}_{i+1}\}.
\end{eqnarray*}
This proves the lemma.
\end{proof}

Fixing $n > 3$, for all $i$ and $j$ ($1 < i < j < n$) define
\begin{eqnarray*}
\Delta^{(n)}_{i,j} & = & \Gamma^{(n)} \setminus \{\delta_i, \delta_j \}\\
\Delta^{*(n)}_{i,j} & = & \Gamma^{*(n)} \setminus \{\delta_i, \delta_j \}.
\end{eqnarray*}
Again $(n)$-superscripts are omitted where possible for clarity.
\begin{lemma}
Let $\theta$ be an $\cR^\dagger$-formula and $\theta^*$ an
$\cR^{*\dagger}$-formula.  For all $i$, $j$ \textup{(}$1 < i < j <
n$\textup{)}, if $\Delta_{i,j} \models \theta$, then $\theta \in
\Delta_{i,j}$. Likewise, if $\Delta^*_{i,j} \models \theta^*$, then
$\theta^* \in \Delta^*_{i,j}$.
\label{lma:DeltaIJ}
\end{lemma}
\begin{proof}
We prove the first statement only; the proof of the second is similar.
Given the equations $\Gamma^{(n)} = \left( {\rm
Th}\left(\fA^{(n)}\right) \setminus \{ \gamma \} \right) \cup \{
\bar{\gamma} \}$ and $\Delta^{(n)}_{i,j} = \Gamma^{(n)} \setminus
\{\delta_i, \delta_j \}$, Lemma~\ref{lma:Bi} yields
\begin{align}
& \fB_{i} \models \Delta_{i,j} \cup \{ \bar{\delta}_i, \delta_j \} 
\label{eq:Bi}\\ 
& \fB_{j} \models \Delta_{i,j} \cup \{ \delta_i, \bar{\delta}_j \}. 
\label{eq:Bj} 
\end{align}
Certainly, then, $\Delta_{i,j}$ is satisfiable; hence, the only
$\theta$ we need consider are those such that neither $\theta$ nor
$\bar{\theta}$ is in $\Delta_{i,j}$.  By the completeness of
$\Gamma^{(n)}$, this entails $\theta \in \{ \delta_i, \bar{\delta}_j,
\bar{\delta}_i, \delta_j \}$. In the first two cases, \eqref{eq:Bi}
shows that $\Delta_{i,j} \not \models \theta$; in the third and fourth
cases, \eqref{eq:Bj} shows the same.
\end{proof}

\begin{lemma}
If $\X$ is a finite set of syllogistic rules in $\cR^\dagger$, 
then there exists $n$ such that, for any absurdity
$\bot$ in $\cR^\dagger$, $\Gamma^{(n)} \not \vdash_\X \bot$.
If $\X$ is a finite set of syllogistic rules in $\cR^{*\dagger}$, 
then there exists $n$ such that, for any absurdity
$\bot$ in $\cR^{*\dagger}$, $\Gamma^{(n)} \not \vdash_\X \bot$.
\label{lma:noEntail}
\end{lemma}
\begin{proof}
We prove the lemma for $\cR^\dagger$; the proof for $\cR^{*\dagger}$
is similar.  Since the number of rules in $\X$ is finite, let $k$ be
the largest number of antecedents in any of these rules, and fix $n =
k+2$.  Since no instance of a rule of $\X$ has more than $n-2$
antecedents, any such set of antecedents included in $\Gamma^{(n)}$
must also be included in $\Delta^{(n)}_{i,j}$ for some $1 < i < j <
n$. Let $\theta$ be the consequent of this rule-instance. Since
$\vdash_\X$ is sound, $\Delta^{(n)}_{i,j} \models \theta$, whence
$\theta \in \Delta^{(n)}_{i,j} \subseteq \Gamma^{(n)}$, by
Lemma~\ref{lma:DeltaIJ}. A simple induction on the lengths of proofs
then shows that any formula derived from $\Gamma^{(n)}$ is in
$\Gamma^{(n)}$. But $\Gamma^{(n)}$ contains no absurdity.
\end{proof}

We now have the promised strengthening of
Corollary~\ref{cor:noSyllR+R*+}.  In some sense,
Lemma~\ref{lma:noEntail} has done all the work; for (RAA) is, in the
context of a $\cR^\dagger$-complete (or $\cR^{*\dagger}$-complete) set
of premises, essentially redundant. 
\begin{theorem}
There exists no finite set $\X$ of syllogistic rules in either $\cR^\dagger$
or $\cR^{*\dagger}$ such that $\Vdash_\X$ is both sound and complete.
\label{theo:R+NoWay}
\end{theorem}
\begin{proof}
We prove the lemma for $\cR^\dagger$; the proof for $\cR^{*\dagger}$
is almost identical.  Suppose $\X$ is a finite set of syllogistic
rules in $\cR^\dagger$, with $\Vdash_\X$ sound.  Let $n$ and
$\Gamma^{(n)}$ be as in Lemma~\ref{lma:noEntail}, 
and let $\bot$ be any absurdity in
$\cR^\dagger$. Thus, by Lemma~\ref{lma:GammaUnsat}, $\Gamma^{(n)} \models \bot$, 
but, by Lemma~\ref{lma:noEntail},  $\Gamma^{(n)}
\not \vdash_\X \bot$. It suffices to show that $\Gamma^{(n)} \not
\Vdash_\X \bot$.

Suppose then, for contradiction, that there is an indirect derivation
of an absurdity from $\Gamma^{(n)}$, using the rules $\X$. Consider
such a derivation in which the number $k$ of applications of (RAA) is
minimal. Since $\Gamma^{(n)} \not \vdash_\X \bot$, we know that $k
>0$. Consider the last application of (RAA) in this
derivation, which derives an $\cR^\dagger$-formula, say, $\bar{\theta}$,
discharging a premise $\theta$. That is, there is an (indirect)
derivation of some absurdity $\bot'$ from 
$\Gamma^{(n)} \cup \{\theta\}$, employing fewer
than $k$ applications of (RAA). 
By minimality of $k$,  $\theta \not \in
  \Gamma^{(n)}$, and so, by the $\cR^\dagger$-completeness of
  $\Gamma^{(n)}$, $\bar{\theta} \in \Gamma^{(n)}$. But then we can
  replace our original derivation of of $\bar{\theta}$ with a trivial
  derivation, so obtaining a derivation of $\bot$ from
  $\Gamma^{(n)}$ with fewer than $k$ applications of (RAA), a contradiction.
\end{proof}

The upshot: unlike in the case of $\cR^*$, the full power of
{\em reductio} does not help with $\cR^\dagger$ and $\cR^{*\dagger}$.

\section{Relation to other work}
\paragraph{Systems related to $\cS$ and $\cS^\dagger$}
Modern treatments of the the syllogistic can all be traced back to
{\L}ukasiewicz~\cite{logic:lukasiewicz}, where a logic is presented in
which formulas of the forms~\eqref{eq:syntaxS} are treated as atoms in
a propositional calculus.  {\L}ukasiewicz provides a collection of
axiom-schemata which, together with the usual axioms of propositional
logic, yields a complete proof-system for the resulting language (a
strict superset of~$\cS$).  The completeness of this system was shown
independently by Westerst{\aa}hl~\cite{westerstahl:89}; note that it
differs from the syllogistic fragments considered in this
paper, in that it is embedded within propositional
logic. Other commentators, for example, Smiley~\cite{logic:smiley73},
Corcoran~\cite{logic:corcoran} and Martin~\cite{logic:martin97},
objecting to {\L}ukasiewicz' exegesis of Aristotle, provide
proof-systems in the form of syllogistic rules similar to those of
$\S$, again proving completeness of their respective systems. These
systems are not {\em direct} syllogistic systems, since they all
employ {\em reductio ad absurdum}.  Our emphasis on direct syllogistic
systems, and especially the formulation of {\em
  refutation-completeness} is new, and motivated by the results we
have obtained on relational extensions of the syllogistic.  In
addition, the authors mentioned above do not treat the fragment
$\cS^\dagger$.  The completeness of $\cS$ itself appears as Theorem
6.2 in Moss~\cite{logic:moss08}, where various fragments of $\cS$ are
also axiomatized.  The completeness of $\S^\dagger$ is proved in
Moss~\cite{logic:moss07}, using a somewhat different treatment.

\paragraph{Systems related to $\cR$}
The first presentation of a complete proof-system for a fragment close
to the relational syllogistic seems to be Nishihara, Morita, and
Iwata~\cite{logic:nishihara+morita90}. This logic is in effect a
relational version of {\L}ukasiewicz', in that formulas roughly
similar to those of the forms~\eqref{eq:syntaxR} are treated as atoms
of a propositional calculus.  The authors provide axiom-schemata
which, together with the usual axioms of propositional logic, yield a
complete proof-system for the language in question.  Actually, the
propositional atoms in this language are allowed to feature $n$-ary
predicates for all $n \geq 1$. However, the rather strange
restrictions on quantifier-scope (existentials must always outscope
universals), mean that this language is primarily of interest for
atoms featuring only unary and binary predicates; these atoms (and
their negations) then essentially correspond the formulas of our
fragment $\cR$. We mention in passing that Nishihara {\em et al.}'s
language includes individual constants; but in practice, this leads to
no useful increase in expressive power, and we ignore this feature.
Because it involves such an expressive fragment, which includes full
propositional logic, their proof certainly yields no upper complexity
bound comparable to that of Theorem~\ref{theo:Rcomplexity}.

A logic inspired by the system of Nishihara~{\em et al} may be found
in Moss~\cite{logic:mosstoappear}.  Roughly speaking, that logic is
the negation-free fragment of $\cR$, corresponding to sentence-forms
involving the words \emph{some} and \emph{all}, but not
\emph{no}. However, it extends $\cR$ in that it allows both readings
of scope-ambiguous sentences. For example, a sentence like {\sf Every
painter admires some artist} has both a subject wide scope reading and
a subject narrow scope reading.  The fact that the subject wide scope
reading of logically implies the subject narrow reading is then a rule
of inference in Moss' system.  No contradictions are possible in the
system.

Moss~\cite{logic:mosstoappear} also analyses a syllogistic logic with
negated \emph{nouns} (not
verbs) and only using \emph{All}.  So in our notation, its formulas
are $\forall(l,m)$ and $\forall(l,\forall(m,r))$, with $l$ and $m$
unary literals.  In addition to rules we have seen, it uses the
following additional rules which might be taken to be forms of the law
of the excluded middle:
$$   \infer{\forall(p,\forall(q,r))}{\forall(n,\forall(q,r)) & \forall(\nbar,\forall(q,r))}
\qquad
\infer{\forall(p,\forall(q,r)).}{\forall(p,\forall(n,r)) & \forall(p,\forall(\nbar,r))}
$$
The system is completed by a  rule with three premises:
$$\infer{\forall(p,\forall(q,r)).}
{\forall(p,\forall(n,r)) &    \forall(o,\forall(q,r)) &  \forall(\obar,\forall(\nbar,r)) }
$$

\paragraph{McAllester and Givan's fragment}
We have already had occasion to mention the results of McAllester and
Givan~\cite{logic:mcA+G92} in connection with our fragment $\cR^*$.
McAllester and Givan present a ``Montagovian syntax'' for a
first-order language over a signature of unary and binary
predicates and individual constants, together with (what they call)
its ``quantifier-free'' fragment.  In fact, this latter fragment is
like our fragment $\cR^*$, except that its `class-terms' (the
equivalent of our c-terms) can be nested to arbitrary depth. Thus, in
McAllester and Givan's language, the formula
\begin{equation*}
\forall(\exists(\exists(\mbox{\sf man},\mbox{\sf kill}),\mbox{\sf kill}),
\exists(\exists(\mbox{\sf animal},\mbox{\sf kill}),\mbox{\sf kill}))
\end{equation*}
expresses the proposition that, as de Morgan might have put it, he who
kills one who kills a man kills one who kills an animal.  However, the
ability to embed c-terms to arbitrary depth is easily seen not to
confer any essential increase in expressive power of sets of formulas,
since deeply nested class-terms can always be `defined out' by
introducing new unary predicates. (This is reflected in the fact that
the satisfiability problem for $\cR^*$ is no different from that of
McAllester and Givan's fragment.)  Interestingly, McAllester and Givan
show that the satisfiability of a set $\Gamma$ of $\cR^*$-formulas can
be determined in polynomial time if, for each class-term $c$ occurring
in $\Gamma$, $\Gamma$ contains either the formula $\exists(c,c)$ or
the formula $\forall(c, \bar{c})$. (Such formula sets are said to {\em
  determine existentials}).  McAllester and Givan provide a
syllogistic-like system for this fragment, which is complete for sets
of formulas which determine existentials. 

A sub-fragment of the  McAllester and
Givan fragment was considered in Moss~\cite{logic:mosstoappear}.
In our terms, the formulas would be of the forms
$\forall(c,d)$ and $\exists(c,d)$, where $c$ and $d$ are
class-terms of the
forms $p$ (an atom), $\exists(c,r)$, or $\forall(d,r)$.
Thus the system lacks negation.
The  rules are  versions of rules we have seen in Section~\ref{sec:noProofR*}:
(T), (I), (B), (D), (J), (K), (L), and (II); it also employs a rule allowing
for reasoning-by-cases,
which is not a syllogistic rule in the sense of this paper.
An informal example shows that the system allows non-trivial inferences,
and inferences with more than one verb.
\begin{equation*}
\infer{\mbox{\sf All who dislike all who respect  all porcupines
 dislike all who  respect all mammals}.}
{\infer{\mbox{\sf All who respect  all mammals respect all porcupines}
 }{\mbox{\sf All porcupines are mammals}
 }
 }
\end{equation*}

\paragraph{Modern revivals of term logic}
In recent decades, various logicians have challenged the dominant
paradigm of quantification-theory by seeking to rehabilitate {\em term
logic}---essentially, the extension of the traditional presentation of
the syllogistic to the case of polyadic relations (see,
e.g.~Sommers~\cite{logic:Sommers82},
Englebretsen~\cite{logic:englebretsen81}). Broadly, the strategy
adopted by these term-logicians has been to stress the {\em
expressiveness} of the new term-logical syntax, and its ability to
represent all that is represented in quantificational logic. (For a
very clear account, see, Michael Lockwood's Appendix G to
Sommers~\cite{logic:Sommers82}, pp.~426--456.) By adopting the new,
term-logical framework---so goes the argument---we obtain a formal
system with all the expressive power of first-order logic, but with
the added (alleged) advantage of greater fidelity to the structure of
natural language.  The outlook of the present paper is rather
different, however.  True, the fragments considered here could be
simply and elegantly expressed using the syntax of term-logic. The
resulting formulas would exhibit only cosmetic differences to the
syntax introduced in Section~\ref{sec:preliminaries}; and, of course,
the associated proof-systems would be, modulo these cosmetic
differences, unaffected. However, we have been at pains to stress the
{\em inexpressiveness} of the fragments studied above, because with
inexpressiveness comes low computational complexity. And from the
point of view of computational complexity, the issue of the ruling
syntactic r\'{e}gime---predicate logic or term logic---is immaterial.

\section{Conclusions}
\label{sec:conclusions}
In this paper, we have investigated the availability of
syllogism-like
proof-systems for various extensions of the traditional syllogistic,
with special emphasis on the need for the rule of {\em reductio ad absurdum}; in
addition, we have derived tight complexity bounds for all the logics
investigated.

These logics are: (i) $\cS$, which corresponds to the traditional
syllogistic; (ii) $\cS^\dagger$, which extends $\cS$ with negated nouns;
(iii) $\cR$, which extends $\cS$ with transitive verbs; (iv) $\cR^\dagger$,
which extends $\cS$ with both these constructions; (v) $\cR^*$, which
extends $\cR$ by allowing subject noun phrases to contain
relative clauses; and (vi) $\cR^{*\dagger}$, which extends $\cR^*$ with
negated nouns. The inclusion relations between these systems, together with
the familiar \emph{two-variable fragment} of first-order logic, are
shown in Figure~\ref{figure-summary}.

The associated table lists these fragments together with the results
we obtained on the existence of syllogistic systems and the complexity
of determining validity of sequents.  Regarding the existence of
syllogistic systems, we showed that: (i) $\cS$ and $\cS^\dagger$ both
have sound and complete direct syllogistic systems (i.e.~systems
containing no rule of {\em reductio ad absurdum}); (ii) $\cR$ has a
sound and {\em refutation-complete} direct syllogistic system
(i.e.~one which becomes complete if {\em reductio ad absurdum} is
allowed as a single, final step), but no sound and complete direct
syllogistic system; (iii) $\cR^*$ has a sound and complete {\em
  indirect} syllogistic system (i.e.~one allowing unrestricted use of
{\em reductio ad absurdum}), but---unless \PTIME$=$\NPTIME---no sound
and refutation-complete direct syllogistic system; (iv) neither
$\cR^\dagger$ nor $\cR^{*\dagger}$ has even an indirect syllogistic
system that is sound and complete. Regarding complexity,
we showed that: (i) the problem of determining the validity of a
sequent in any of $\cS$, $\cS^\dagger$ or $\cR$ is
\NLOGSPACE-complete; (ii) the problem of determining the validity of a
sequent in $\cR^*$ is co-\NPTIME-complete; (iii) the problem of determining
the validity of a sequent in either of $\cR^\dagger$ or $\cR^{*\dagger}$ is
\EXPTIME-complete. 
\begin{figure}
\begin{minipage}{10cm}
\begin{minipage}{3cm}
\begin{picture}(60,80)
\put(40,0){$\cS$}
\put(20,20){$\cR$}
\put(0,40){$\cR^*$}
\put(60,20){$\cS^\dagger$}
\put(40,40){$\cR^\dagger$}
\put(20,60){$\cR^{*\dagger}$}
\put(20,80){$\FO^2$}

\put(48,8){\rotatebox{45}{$\subseteq$}}
\put(28,28){\rotatebox{45}{$\subseteq$}}
\put(8,48){\rotatebox{45}{$\subseteq$}}

\put(28,12){\rotatebox{315}{$\supseteq$}}
\put(8,32){\rotatebox{315}{$\supseteq$}}

\put(48,32){\rotatebox{315}{$\supseteq$}}
\put(28,52){\rotatebox{315}{$\supseteq$}}

\put(22,70){\rotatebox{90}{$\subseteq$}}

\end{picture}
\end{minipage}
\begin{minipage}{6cm}
{\small 
\begin{tabular}{|l|l|l|}
\hline
$\cS$ &direct, complete  & \NLOGSPACE\\
$\cS^\dagger$ & direct, complete   &   \NLOGSPACE \\
$\cR$ &  direct, refutation complete   & \NLOGSPACE  \\
$\cR^\dagger$ &  not even indirect  & \EXPTIME\\
$\cR^*$ & indirect, complete & Co-\NPTIME~\cite{logic:mcA+G92}  \\ 
$\cR^{*\dagger}$ &  not even indirect & \EXPTIME \\
$\FO^2$ &  & \NEXPTIME~\cite{logic:gkv97} \\
\hline
\end{tabular}
}
\end{minipage}
\end{minipage}
\caption{The six fragments studied in this paper
together with the two-variable fragment $\FO^2$ of
first-order logic. The table shows {\em
    strongest possible} results on the existence of syllogistic
  systems, together with tight complexity bounds. See
  Section~\ref{sec:conclusions} for an explanation.}
\label{figure-summary}
\end{figure}
\bibliographystyle{plain} 
\bibliography{relsyll}
\end{document}